\documentclass[a4paper,UKenglish,british]{lipics-v2018}

\hideLIPIcs

\usepackage{pdfsync}
\usepackage{gastex}
\usepackage{microtype}%if unwanted, comment out or use option "draft"

\usepackage{todonotes}
\title{Languages ordered by the subword order}

\titlerunning{Languages ordered by the subword order -- long version}%optional, please use if title is longer than one line

\author{Dietrich Kuske}{TU Ilmenau, Germany}{dietrich.kuske@tu-ilmenau.de}{}{}%mandatory, please use full name; only 1 author per \author macro; first two parameters are mandatory, other parameters can be empty.

\author{Georg Zetzsche}{MPI-SWS, Germany}{georg@mpi-sws.org}{}{}

\authorrunning{D. Kuske and G. Zetzsche}%mandatory. First: Use abbreviated first/middle names. Second (only in severe cases): Use first author plus 'et al.'

\Copyright{D. Kuske and G. Zetzsche}%mandatory, please use full first names. LIPIcs license is "CC-BY";  http://creativecommons.org/licenses/by/3.0/

\subjclass{,~
}% mandatory: Please choose ACM 2012 classifications from https://www.acm.org/publications/class-2012 or https://dl.acm.org/ccs/ccs_flat.cfm . E.g., cite as "General and reference $\rightarrow$ General literature" or \ccsdesc[100]{General and reference~General literature}. 

%mandatory

\category{}%optional, e.g. invited paper

\relatedversion{}%optional, e.g. full version hosted on arXiv, HAL, or other respository/website

\supplement{}%optional, e.g. related research data, source code, ... hosted on a repository like zenodo, figshare, GitHub, ...

\funding{}%optional, to capture a funding statement, which applies to all authors. Please enter author specific funding statements as fifth argument of the \author macro.

\acknowledgements{}%optional

\EventEditors{John Q. Open and Joan R. Access}
\EventNoEds{2}
\EventLongTitle{42nd Conference on Very Important Topics (CVIT 2016)}
\EventShortTitle{CVIT 2016}
\EventAcronym{CVIT}
\EventYear{2016}
\EventDate{December 24--27, 2016}
\EventLocation{Little Whinging, United Kingdom}
\EventLogo{}
\SeriesVolume{42}
\ArticleNo{23}
\nolinenumbers %uncomment to disable line numbering
\theoremstyle{theorem}
\newtheorem{proposition}[theorem]{Proposition}

\newcommand{\embedsright}{\hookrightarrow}
\newcommand{\CMOD}{\ensuremath{\mathrm{C{+}MOD}}}
\newcommand{\FOMOD}{\ensuremath{\mathrm{FO{+}MOD}}}
\newcommand{\FO}{\ensuremath{\mathrm{FO}}}
\newcommand{\Inc}{\mathrm{Inc}}
\newcommand{\Sub}{\mathrm{Sub}}
\newcommand{\bN}{\mathbb{N}}
\newcommand{\bZ}{\mathbb{Z}}
\newcommand{\cS}{\mathcal{S}}

\newcommand{\proj}{\mathrm{proj}}
\newcommand{\propersubword}{\sqsubset}

\newcommand{\subword}{\sqsubseteq}

\newcommand{\cover}{\mathbin{\propersubword\!\!\!\!\cdot}}

\begin{document}

\maketitle

\begin{abstract}
  We consider a language together with the subword relation, the cover
  relation, and regular predicates. For such structures, we consider
  the extension of first-order logic by threshold- and modulo-counting
  quantifiers. Depending on the language, the used predicates, and the
  fragment of the logic, we determine four new combinations that yield
  decidable theories. These results extend earlier ones where only the
  language of all words without the cover relation and fragments of
  first-order logic were considered.
\end{abstract}

\section{Introduction}

The subword relation (sometimes called scattered subword relation) is
one of the simplest nontrivial examples of a well-quasi ordering
\cite{Hig52}. This property allows its prominent use in the
verification of infinite state systems \cite{FinS01}. The subword
relation can be understood as embeddability of one word into
another. This embeddability relation has been considered for other
classes of structures like trees, posets, semilattices, lattices,
graphs
etc.~\cite{KudS07,KudS09,KudSY10,JezM09,JezM09a,JezM09b,JezM10,Thi17,Thi18}.

In this paper, we study logical properties of a set of words ordered
by the subword relation. We are mainly interested in general
situations where we get a decidable logical theory. Regarding
first-order logic, we already have a rather precise picture about the border
between decidable and undecidable fragments: For the subword order
alone, the $\exists^*$-theory is decidable \cite{Kus06} and the
$\exists^*\forall^*$-theory is undecidable \cite{KarS15}.  For the
subword order together with regular predicates, the two-variable
theory is decidable \cite{KarS15} and the three-variable theory
\cite{KarS15} as well as the $\exists^*$-theory are undecidable
\cite{HalSZ17} (these two undecidabilities already hold if we only
consider singleton predicates, i.e., constants).

Thus, to get a decidable theory, one has to restrict the
expressiveness of first-order logic considerably. For instance,
neither in the $\exists^*$-, nor in the two-variable fragment of
first-order logic, one can express the cover relation $\cover$ (i.e.,
``$u$ is a proper subword of $v$ and there is no word properly between
these two''). As another example, one cannot express threshold
properties like ``there are at most $k$ subwords with a given
property'' in any of these two logics (for $k>2$).

In this paper, we refine the analysis of logical properties of the
subword order in three aspects:
\begin{itemize}
\item We restrict the universe from the set of all words to a given
  language $L$.
\item Besides the subword order, we also consider the cover relation
  $\cover$.
\item We add threshold and modulo counting quantifiers to the logic.
\end{itemize}
Also as before, we may or may not add regular predicates or constants
  to the structure.
In other words, we consider reducts of the structure
\[
  (L,\subword,\cover,(K\cap L)_{K\text{ regular}},(w)_{w\in L})
\]
with $L$ some language
and fragments of the logic $\CMOD$ that extends first-order logic by
threshold- and modulo-counting quantifiers.

In this spectrum, we identify four new cases of decidable theories:
\begin{enumerate}
\item The $\CMOD$-theory of the whole structure is decidable provided
  $L$ is bounded and context-free (Theorem~\ref{T-bounded}). A rather
  special case of this result follows from
  \cite[Theorem~4.1]{HalSZ17}: If $L=(a_1^*a_2^*\cdots a_m^*)^\ell$
  and if only regular predicates of the form
  $(a_1^*a_2^*\cdots a_m^*)^k$ are used, then the $\FO$-theory is
  decidable (with $\Sigma=\{a_1,\dots,a_m\}$). 
\item The $\CMOD^2$-theory (i.e., the 2-variable-fragment of the
  $\CMOD$-theory) of the whole structure is decidable whenever $L$ is
  regular (Corollary~\ref{C-CMOD2}). The decidability of the
  $\FO^2$-theory without the cover relation is
  \cite[Theorem~5.5]{KarS15}.
\item The $\Sigma_1$-theory of the structure $(L,\subword)$ is
  decidable provided $L$ is regular (Theorem~\ref{T-unbounded}). For
  $L=\Sigma^*$, this is \cite[Prop.~2.2]{Kus06}.
\item The $\Sigma_1$-theory of the structure
  $(L,\subword,(w)_{w\in L})$ is decidable provided $L$ is regular and
  almost every word from $L$ contains a non-negligible number of
  occurrences of every letter (see below for a precise definition,
  Theorem~\ref{T-nice}). Note that, by \cite[Theorem~3.3]{HalSZ17},
  this theory is undecidable if $L=\Sigma^*$.
\end{enumerate}

Our first result is shown by an interpretation of the structure in
$(\bN,+)$. Four ingredients are essential here: Parikh's
theorem~\cite{Par66}, the rationality of the subword
relation~\cite{KarS15}, Nivat's theorem characterising rational
relations~\cite{Ber79}, and the decidability of the $\CMOD$-theory of
$(\bN,+)$~\cite{Ape66,Sch05,HabK15} (note that this decidability does
not follow directly from Presburger's result since in his logic, one
cannot make statements like ``the number of witnesses $x\in\bN$
satisfying \dots is even'').

Our second result extends a result from~\cite{KarS15} that shows decidability of the $\FO^2$-theory of the structure
$(\Sigma^*,\subword,(L)_{L\text{ regular}})$ . They
provide a quantifier elimination procedure which relies on two facts:
\begin{itemize}
\item The class of regular languages is closed under images of
  rational relations.
\item The proper subword relation and the incomparability relation are
  rational.
\end{itemize}
Here, we follow a similar proof strategy. But while that proof had to
handle the existential quantifier, only, we also have to deal with
counting quantifiers. This requires us to develop a theory of counting
images under rational relations, e.g., the set of words $u$ such that
there are at least two words $v$ in the regular language $K$ with
$(u,v)$ in the rational relation $R$. We show that the class of
regular languages is closed under such counting images provided the
rational relation $R$ is unambiguous, a proof that makes heavy use of
weighted automata~\cite{Sak09}. To apply this to the subword and the
cover relation, this also requires us to show that the proper subword,
the cover, and the incomparability relations are unambiguous rational.

Our third result extends the decidability of the $\Sigma_1$-theory of
$(\Sigma^*,\subword)$ from~\cite{Kus06}. The main point there was to
prove that every finite partial order can be embedded into
$(\Sigma^*,\subword)$ if $|\Sigma|\ge 2$. This is certainly false if we restrict the
universe, e.g., to $L=a^*$. However, such bounded regular languages are
already covered by the first result, so we only have to handle
unbounded regular languages~$L$. In that case, we prove nontrivial
combinatorial results regarding primitive words in regular languages
and prefix-maximal subwords. These considerations then allow us to
prove that, indeed, every finite partial order embeds into
$(L,\subword)$. Then, the decidability of the $\Sigma_1$-theory
follows as in \cite{Kus06}.

Regarding our fourth result, we know from~\cite{KarS15} that
decidability of the $\Sigma_1$-theory of $(L,\subword,(w)_{w\in L})$
does not hold for every regular $L$.  Therefore, we require that a
certain fraction of the positions in a word carries the letter $a$
(for almost all words from the language and for all letters). This
allows us to conclude that every finite partial order embeds into
$(L,\subword)$ above each word. The second ingredient is that,
for such languages, any $\Sigma_1$-sentence is effectively equivalent
to such a sentence where  constants are only used to express that
all variables take values above a certain word~$w$. These two
properties together with some combinatorial arguments from the theory
of well-quasi orders then yield the decidability.

In summary, we identify four classes of decidable theories related to
the subword order. In this paper, we concentrate on these positive
results, i.e., we did not try to find new undecidable theories. It
would, in particular, be nice to understand what properties of the
regular language $L$ determine the decidability of the
$\Sigma_1$-theory of the structure $(L,\subword,(w)_{w\in L})$ (it is
undecidable for $L=\Sigma^*$ \cite{HalSZ17} and decidable for, e.g.,
$L=\{ab,baa\}^*\cup bb\{abb\}^*$ by our third result). Another open
question concerns the complexity of our decidability results.

\section{Preliminaries}

Throughout this paper, let $\Sigma$ be some alphabet. A word
$u=a_1a_2\dots a_m$ with $a_1,a_2,\dots,a_m\in\Sigma$ is a
\emph{subword} of a word $v\in\Sigma^*$ if there are words
$v_0,v_1,\dots,v_m\in\Sigma^*$ with $v=v_0a_1v_1a_2v_2\cdots
a_mv_m$. In that case, we write $u\subword v$; if, in addition,
$u\neq v$, then we write $u\propersubword v$ and call $u$ a \emph{proper}
subword of $v$. If $u,w\in\Sigma^*$ such that $u\propersubword w$ and
there is no word $v$ with $u\propersubword v\propersubword w$, then we
say that $w$ is a \emph{cover} of $u$ and write $u\cover w$. This is
equivalent to saying $u\subword w$ and $|u|+1=|w|$ where $|u|$ is the
length of the word $u$.  If, for two words $u$ and $v$, neither $u$ is
a subword of $v$ nor \textit{vice versa}, then the words $u$ and $v$
are \emph{incomparable} and we write $u\parallel v$. For instance,
$aa\propersubword babbba$, $aa\cover aba$, and $aba\parallel aabb$.
\medskip

Let $\mathcal S=(L,(R_i)_{i\in I},(w_j)_{j\in J})$ be a
\emph{structure}, i.e., $L$ is a set, $R_i\subseteq L^{n_i}$ is a
relation of arity~$n_i$ (for all $i\in I$), and $w_j\in L$ for all
$j\in J$. Then, formulas of the logic $\CMOD$ are built from the
atomic formulas
\begin{itemize}
\item $s=t$ for $s,t$ variables or constants $w_j$ and
\item $R_i(s_1,s_2,\dots,s_{n_i})$ for $i\in I$ and
  $s_1,s_2,\dots,s_{n_i}$ variables or constants $w_j$
\end{itemize}
by the following formation rules:
\begin{enumerate}
\item If $\alpha$ and $\beta$ are formulas, then so are $\lnot\alpha$
  and $\alpha\land\beta$.
\item If $\alpha$ is a formula and $x$ a variable, then
  $\exists x\,\alpha$ is a formula.
\item If $\alpha$ is a formula, $x$ a variable, and $k\in\bN$, then
  $\exists^{\ge k} x\,\alpha$ is a formula.
\item If $\alpha$ is a formula, $x$ a variable, and $p,q\in\bN$ with
  $p<q$, then $\exists^{p\bmod q} x\,\alpha$ is a formula.
\end{enumerate}
We call $\exists^{\ge k}$ a \emph{threshold counting quantifier} and
$\exists^{p\bmod q}$ a \emph{modulo counting quantifier}. The
semantics of these quantifiers is defined as follows:
\begin{itemize}
\item $\mathcal S\models\exists^{\ge k}x\,\alpha$ iff
  $|\{w\in L\mid \mathcal S \models\alpha(w)\}|\ge k$
\item $\mathcal S\models\exists^{p\bmod q}x\,\alpha$ iff
  $|\{w\in L\mid \mathcal S \models\alpha(w)\}|\in p+q\bN$
\end{itemize}
For instance, $\exists^{0\bmod 2}x\,\alpha$ expresses that the number
of elements of the structure satisfying $\alpha$ is even. Then
$\bigl(\exists^{0\bmod 2}x\,\alpha\bigr)\lor\bigl(\exists^{1\bmod
  2}x\,\alpha\bigr)$ holds iff only finitely many elements of the
structure satisfy $\alpha$. The fragment $\FOMOD$ of $\CMOD$ comprises
all formulas not containing any threshold counting quantifier
$\exists^{\ge k}$. First-order logic $\FO$ is the set of formulas from
$\CMOD$ not mentioning any counting quantifier, i.e., neither
$\exists^{\ge k}$ nor $\exists^{p\bmod q}$. Let $\Sigma_1$ denote the
set of first-order formulas of the form
$\exists x_1\,\exists x_2\dots \exists x_n\colon\psi$ where $\psi$ is
quantifier-free; these formulas are also called \emph{existential}.

Note that the formulas
\begin{equation}
  \exists^{\ge k} x\,x=x\label{eq:C}
\end{equation}
and
\begin{equation}
  \label{eq:FO}
  \exists x_1\,\exists x_2\dots\exists x_k\colon
  \bigwedge_{1\le i<j\le k}x_i\neq x_j\land\bigwedge_{1\le i\le k}x_i=x_i
\end{equation}
are equivalent (they both express that the structure contains at least
$k$ elements). Generalising this, the threshold quantifier
$\exists^{\ge k}$ can be expressed using the existential quantifier,
only. Consequently, the logics $\FOMOD$ and $\CMOD$ are equally
expressive. The situation changes when we restrict the number of
variables that can be used in a formula: Let $\FO^2$ and $\CMOD^2$
denote the set of formulas from $\FO$ and $\CMOD$, respectively, that
use the variables $x$ and $y$, only.  Note that the formula from
\eqref{eq:C} belongs to $\CMOD^2$, but the equivalent formula from
\eqref{eq:FO} does not belong to $\FOMOD^2$.

\begin{remark}
  Let $L_1$ be a set with two elements and let $L_2$ be a set with
  $k\cdot q+2$ elements (where $k>2$). Furthermore, let
  $\varphi\in\CMOD^2$ be a formula such that all moduli appearing in
  $\varphi$ are divisors of $q$. By induction on the construction of
  the formula $\varphi$, one can show the following for any
  $x,y\in L_1$ and $x',y'\in L_2$:
  \begin{itemize}
  \item If $x\neq y$ and $x'\neq y'$, then
    $L_1\models\varphi(x,y)\iff L_2\models\varphi(x',y')$.
  \item $L_1\models\varphi(x,x)\iff L_2\models\varphi(x',x')$.
  \end{itemize}
  Consequently, there is no $\CMOD^2$-formula expressing that the
  number of elements of a structure is $\ge k$.
\end{remark}

In this paper, we will consider the following structures:
\begin{itemize}
\item The largest one is
  $(L,\subword,\cover,(K\cap L)_{L\text{ regular}},(w)_{w\in L})$ for
  some $L\subseteq\Sigma^*$. The universe of this structure is the
  language~$L$, we have two binary predicates ($\subword$ and
  $\cover$), a unary predicate $K\cap L$ for every regular language
  $L$, and we can use every word from $L$ as a constant.
\item The other extreme is the structure $(L,\subword)$ for some
  $L\subseteq\Sigma^*$ where we consider only the binary predicate
  $\subword$.
\item Finally, we will also prove results on the intermediate
  structure $(L,\subword,(w)_{w\in L})$ that has a binary relation and
  any word from the language as a constant.
\end{itemize}
For any structure $\mathcal S$ and any of the above logics
$\mathcal L$, we call
\[
  \{\varphi\mid \varphi\in \mathcal L\text{ sentence and } \mathcal S\models\varphi\}
\]
the \emph{$\mathcal L$-theory of $\mathcal S$}.
\medskip

A language $L\subseteq\Sigma^*$ is \emph{bounded} if there are a
number $n\in\bN$ and words $w_1,w_2,\dots,w_n\in\Sigma^*$ such that
$L\subseteq w_1^*\,w_2^*\cdots w_n^*$. Otherwise, it is
\emph{unbounded}.

For an alphabet $\Gamma$, a word $w\in\Gamma^*$, and a letter
$a\in\Gamma$, let $|w|_a$ denote the number of occurrences of the
letter $a$ in the word $w$. The \emph{Parikh vector} of $w$ is the
tuple $\Psi_\Gamma(w)=(|w|_a)_{a\in\Gamma}\in\bN^\Gamma$. Note that
$\Psi_\Gamma$ is a homomorphism from the free monoid $\Gamma^*$ onto
the additive monoid $(\bN^\Gamma,+)$.

\section{The $\FOMOD$-theory with regular predicates}

The aim of this section is to prove that the full $\FOMOD$-theory of
the structure
\[
  (L,\subword,(K\cap L)_{K\text{ regular}})
\]
is decidable for $L$ bounded and context-free. This is achieved by
interpreting this structure in $(\bN,+)$, i.e., in Presburger
arithmetic whose $\FOMOD$-theory is known to be
decidable~\cite{Ape66,Sch05,HabK15}.

We start with three preparatory lemmas.

\begin{lemma}\label{L1}
  Let $K\subseteq\Sigma^*$ be context-free,
  $w_1,\dots,w_n\in\Sigma^*$, and $g\colon\bN^n\to\Sigma^*$ be defined
  by $g(\overline{m})=w_1^{m_1}\,w_2^{m_2}\cdots w_n^{m_n}$ for all
  $\overline{m}=(m_1,m_2,\dots,m_n)\in\bN^n$. The set
  $g^{-1}(K)=\{\overline{m}\in\bN^n\mid g(\overline{m})\in K\}$ is
  effectively semilinear.
\end{lemma}

\begin{proof}
  Let $\Gamma=\{a_1,a_2,\dots,a_n\}$ be an alphabet and define the
  monoid homomorphism $f\colon\Gamma^*\to\Sigma^*$ by $f(a_i)=w_i$ for
  all $i\in[1,n]$.
  
  Since the class of context-free languages is effectively closed
  under inverse homomorphisms, the language
  \[
    K_1=f^{-1}(K)=\{u\in\Gamma^*\mid f(u)\in K\}
  \]
  is effectively context-free. Since $a_1^*a_2^*\dots a_n^*$ is
  regular, also the language
  \[
    K_2=K_1\cap a_1^*a_2^*\dots a_n^*
  \]
  is effectively context-free. By Parikh's theorem~\cite{Par66}, the
  Parikh-image $\Psi_\Gamma(K_2)\subseteq\bN^n$ of this intersection
  is effectively semilinear.

  Now let $\overline{m}\in \bN^n$. Then
  $\overline{m}\in \Psi_\Gamma(K_2)$ iff there exists a word
  $u\in K_1\cap a_1^*a_2^*\dots a_n^*$ with Parikh image
  $\overline{m}$. But the only word from $a_1^*a_2^*\dots a_n^*$ with
  this Parikh image is $a_1^{m_1}a_2^{m_2}\cdots a_n^{m_n}$, i.e.,
  $\overline{m}\in\Psi_\Gamma(K_2)$ iff
  $a_1^{m_1}a_2^{m_2}\cdots a_n^{m_n}\in K_1$.  Since
  $f(a_1^{m_1}a_2^{m_2}\cdots a_n^{m_n})=g(\overline{m})$, this is
  equivalent to $g(\overline{m})\in K$. Thus, the semilinear set
  $\Psi_\Gamma(K_2)$ equals the set $g^{-1}(K)$ from the lemma.
\end{proof}

\begin{lemma}\label{L2}
  Let $w_1,\dots,w_n\in\Sigma^*$ and $g\colon\bN^n\to\Sigma^*$ be
  defined by $g(\overline{m})=w_1^{m_1}\,w_2^{m_2}\cdots w_n^{m_n}$
  for all $\overline{m}=(m_1,m_2,\dots,m_n)\in\bN^n$. The set
  $\{(\overline{m},\overline{n})\in\bN^n\times\bN^n\mid
  g(\overline{m})\sqsubseteq g(\overline{n})\}$ is
  semilinear.
\end{lemma}

\begin{proof}
  Let $\Gamma=\{a_1,a_2,\dots,a_n\}$ be an alphabet and define the
  monoid homomorphism $f\colon\Gamma^*\to\Sigma^*$ by $f(a_i)=w_i$ for
  all $i\in[1,n]$.

  In this prove, we construct an alphabet $\Delta$ and homomorphisms
  $g$, $h_1$, $h_2$, $p_1$, and $p_2$, such that the diagrams (for
  $i\in\{1,2\}$) from Fig.~\ref{F} commute. In addition, we will
  construct a regular language $R\subseteq\Delta^*$ with
  \[
    U\subword V\iff\exists w\in R\colon U=f\circ h_1(w)\text{ and }V=f\circ h_2(w)
  \]
  for all $U,V\in\Sigma^*$.
  
  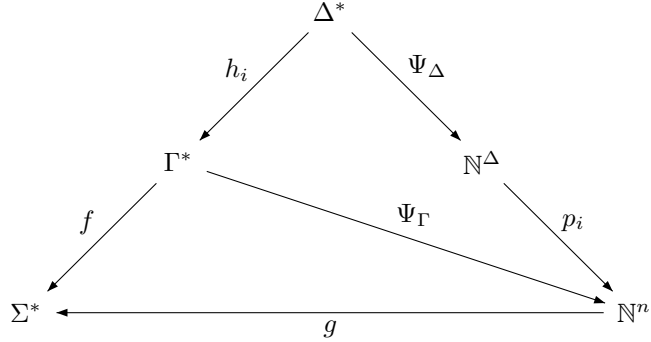
\begin{figure}
    \centering
    \begin{picture}(100,60)
      \gasset{Nframe=n}
      \node(1)(50,50){$\Delta^*$}
      \node(2)(30,30){$\Gamma^*$}
      \node(3)(10,10){$\Sigma^*$}
      \node(4)(70,30){$\bN^\Delta$}
      \node(5)(90,10){$\bN^n$}
      \drawedge[ELside=r](1,2){$h_i$}
      \drawedge[ELside=r](2,3){$f$}
      \drawedge(1,4){$\Psi_\Delta$}
      \drawedge(4,5){$p_i$}
      \drawedge(2,5){$\Psi_\Gamma$}
      \drawedge(5,3){$g$}
    \end{picture}
  
    \caption{Commuting diagram for proof of Lemma~\ref{L2}}
    \label{F}
  \end{figure}

  The subword relation
  \[
    S=\{(U,V)\in\Sigma^*\times\Sigma^*\mid U\subword V\}
  \]
  on $\Sigma^*$ is rational~\cite{KarS15}. Since the class of rational
  relations is closed under inverse homomorphisms \cite{Ber79}, also
  the relation
  \[
    S_1=\{(u,v)\in\Gamma^*\times\Gamma^*\mid f(u)\subword f(v)\}
  \]
  is rational. While the class of rational relations is not closed
  under intersections, it is at least closed under intersections with
  direct products of regular languages. Hence also
  \[
    S_2=\{(u,v)\mid u,v\in a_1^*a_2^*\dots a_n^*,f(v)\subword f(v)\}
  \]
  is rational. By Nivat's theorem \cite{Ber79}, there are a regular
  language $R$ over some alphabet $\Delta$ and two homomorphisms
  $h_1,h_2\colon\Delta^*\to\Gamma^*$ with
  \[
    S_2=\{(h_1(w),h_2(w))\mid w\in R\}\,.
  \]
  Since $R$ is regular, its Parikh-image
  \[
    \Psi_\Delta(R)=\{\Psi_\Delta(w)\mid w\in R\}
  \]
  is semilinear~\cite{Par66}. 

  Note that the additive monoid $(\bN^\Delta,+)$ is the commutative
  monoid freely generated by the vectors
  $\Psi_\Delta(b)=(0,\dots,0,1,0,\dots,0)$ for $b\in\Delta$. Since
  also $(\bN^n,+)$ is commutative, we can define monoid homomorphisms
  $p_1,p_2\colon\bN^{\Delta}\to\bN^n$ by
  $p_1(\Psi_\Delta(b))=\Psi_\Gamma(h_1(b))$ and
  $p_2(\Psi_\Delta(b))=\Psi_\Gamma(h_2(b))$ for $b\in\Delta$.  Then
  also
  $(p_1,p_2)\colon\bN^\Delta\to\bN^n\times\bN^n\colon
  \overline{x}\mapsto \bigl(p_1(\overline{x}),p_2(\overline{x})\bigr)$
  is a monoid homomorphism.

  Let $w=b_1b_2\cdots b_m$ with $b_i\in\Delta$ for all $1\le i\le
  m$. Then we have
  \begin{align*}
    \Psi_\Gamma(h_1(w))
    &=\Psi_\Gamma(h_1(b_1b_2\cdots b_m))\\
    &=\Psi_\Gamma(h_1(b_1)h_1(b_2)\cdots h_1(b_m))\\
    &=\sum_{1\le j\le m}\Psi_\Gamma(h_1(b_j))\\
    &=\sum_{1\le j\le m} p_1(\Psi_\Delta(b_j))\\
    &=p_1(\sum_{1\le j\le m}\Psi_\Delta(b_j))\\
    &=p_1(\Psi_\Delta(w))
  \end{align*}
  and similarly $\psi_\Gamma(h_2(w))=p_2(\Psi_\Delta(w))$.

  Since the class of semilinear sets is closed under monoid
  homomorphisms, the image of the semilinear set $\Psi_\Delta(R)$
  under $(p_1,p_2)$, i.e.,
  \[
    H=\{\bigl(p_1(\Psi_\Delta(w)),p_2(\Psi_\Delta(w))\bigr)\mid w\in R\}
  \]
  is semilinear. 

  Let $\overline{m},\overline{n}\in\bN^n$. Then
  $(\overline{m},\overline{n})\in H$ iff there exists $w\in R$ with
  $\overline{m}=p_1(\Psi_\Delta(w))=\Psi_\Gamma(h_1(w))$ and
  $\overline{n}=p_2(\Psi_\Delta(w))=\Psi_\Gamma(h_2(w))$. The
  existence of such a word $w\in R$ is equivalent to the existence of
  a pair $(u,v)\in S_2$ with $\overline{m}=\Psi_\Gamma(u)$ and
  $\overline{n}=\Psi_\Gamma(v)$.

  Since all words appearing in $S_2$ belong to
  $a_1^*a_2^*\cdots a_n^*$, this last statement is equivalent to saying
  \[
    (a_1^{m_1}a_2^{m_2}\cdots a_n^{m_n},a_1^{n_1}a_2^{n_2}\cdots a_n^{n_n})\in S_2\,.
  \]
  But this is equivalent to saying
  \[
    f(a_1^{m_1}a_2^{m_2}\cdots a_n^{m_n}) \subword
    f(a_1^{n_1}a_2^{n_2}\cdots a_n^{n_n})\,.
  \]
  Note that $g(\overline{m})=f(a_1^{m_1}a_2^{m_2}\cdots a_n^{m_n})$
  and similarly
  $g(\overline{n})=f(a_1^{n_1}a_2^{n_2}\cdots a_n^{n_n})$. Thus, the
  last claim is equivalent to
  $g(\overline{m})\subword g(\overline{n})$.

  In summary, we showed that the semilinear set $H$ is the set from
  the lemma.
\end{proof}

\begin{lemma}\label{L3}
  Let $w_1,w_2,\dots,w_n\in\Sigma^*$,
  $L\subseteq w_1^*w_2^*\cdots w_n^*$ be context-free, and
  $g\colon\bN^n\to\Sigma^*$ be defined by
  $g(\overline{m})=w_1^{m_1}w_2^{m_2}\cdots w_n^{m_n}$ for every tuple
  $\overline{m}=(m_1,m_2,\dots,m_n)\in\bN^n$.
  Then there exists a semilinear set $U\subseteq\bN^n$ such that $g$
  maps $U$ bijectively onto $L$.
\end{lemma}

\begin{proof}
  By Lemma~\ref{L1}, the set $g^{-1}(L)$ is semilinear and satisfies,
  by its definition, $g(g^{-1}(L))=L\cap w_1^*w_2^*\cdots w_n^*=L$.

  Let $T$ denote the semilinear set from Lemma~\ref{L2}.

  Then let $U$ denote the set of $n$-tuples
  $\overline{m}\in g^{-1}(L)$ such that the following holds for all
  $\overline{n}\in \bN^n$:

  If $(\overline{m},\overline{n})\in T$ and
  $(\overline{n},\overline{m})\in T$, then $\overline{m}$ is
  lexicographically smaller than or equal to $\overline{n}$.

  This set $U$ is semilinear since the class of semilinear relations
  is closed under first-order definitions. Now let $u\in L$. Since $g$
  maps $g^{-1}(L)$ onto $L$, there is $\overline{m}\in g^{-1}(L)$ with
  $g(\overline{m})=u$. Since, on $g^{-1}(L)\subseteq\bN^n$, the
  lexicographic order is a well-order, there is a lexicographically
  minimal such tuple $\overline{m}$. This tuple belongs to $U$ and it
  is the only tuple from $U$ mapped to $u$.
\end{proof}

Now we can prove the main result of this section.

\begin{theorem}\label{T-bounded}
  Let $L\subseteq\Sigma^*$ be context-free and bounded. Then the
  $\FOMOD$-theory of $\cS=(L,\subword,(K\cap L)_{K\text{ regular}})$
  is decidable.
\end{theorem}

\begin{proof}
  Since $L$ is bounded, there are words $w_1,w_2,\dots,w_n\in\Sigma^*$
  such that $L\subseteq w_1^*\,w_2^*\cdots w_n^*$. For an $n$-tuple
  $\overline{m}=(m_1,m_2,\dots,m_n)\in\bN^n$ we define
  $g(\overline{m})= w_1^{m_1}w_2^{m_2}\cdots w_n^{m_n}\in\Sigma^*$.

  \begin{enumerate}
  \item By Lemma~\ref{L3}, there is a semilinear set $U\subseteq\bN^n$
    that is mapped by $g$ bijectively onto $L$.

    From this semilinear set, we obtain a first-order formula
    $\lambda(\overline{x})$ in the language of $(\bN,+)$ such that,
    for any $\overline{m}\in\bN^n$, we have
    $(\bN,+)\models\lambda(\overline{m})\iff\overline{m}\in U$.
  \item The set
    $\{(\overline{m},\overline{n})\mid g(\overline{m})\subword
    g(\overline{n})\}$ is semilinear by Lemma~\ref{L2}.

    From this semilinear set, we obtain a first-order formula
    $\sigma(\overline{x},\overline{y})$ in the language of $(\bN,+)$
    such that
    $(\bN,+)\models\sigma(\overline{m},\overline{n})\iff
    g(\overline{m})\subword g(\overline{n})$.
  \item For any regular language $K\subseteq\Sigma^*$ the set
    $\{\overline{m}\in\bN^n\mid g(\overline{m})\in K\}\subseteq\bN^n$
    is effectively semilinear by Lemma~\ref{L1}.

    From this semilinear set, we can compute a first-order formula
    $\kappa_K(\overline{x})$ in the language of $(\bN,+)$ such that
    $(\bN,+)\models\kappa_K(\overline{m})\iff g(\overline{m})\in K$.
  \end{enumerate}

  We now define, from an $\FOMOD$-formula $\varphi(x_1,\dots,x_k)$ in
  the language of $\cS$, an $\FOMOD$-formula
  $\varphi'(\overline{x_1},\dots,\overline{x_k})$ in the language of
  $(\bN,+)$ such that
  \[
    (\bN,+)\models\varphi'(\overline{m_1},\dots,\overline{m_k})
    \iff
    \cS\models\varphi(g(\overline{m_1}),\dots,g(\overline{m_k}))\,.
  \]
  If $\varphi=(x\subword y)$, then
  $\varphi'=\sigma(\overline{x},\overline{y})$. If $\varphi=(x\in K)$,
  then $\varphi'=\kappa_K(\overline{x})$. Furthermore,
  $(\alpha\land\beta)'=\alpha'\land\beta'$ and
  $(\lnot\varphi)'=\lnot\varphi'$.

  For $\varphi=\exists x\colon\psi$, we set
  $\varphi'=\exists x_1\exists x_2\dots \exists x_n\colon
  \lambda(x_1,x_2,\dots,x_n)\land\psi'$.

  Finally, let $\varphi=\exists^{p\bmod q}x\colon\psi$. Intuitively,
  one is tempted to set
  $\varphi'=\exists^{p\bmod
    q}\overline{x}\colon\lambda(\overline{x})\land\psi'$, but this is
  not a valid formula since $\overline{x}$ is not a single variable,
  but a tuple of variables. To rectify this, we define
  $\FOMOD$-formulas $\alpha_p^k$ for $p\in[0,q-1]$ and $k\in[0,n-1]$
  as follows:
  \[
    \alpha_p^k(x_1,\dots,x_k)=
    \begin{cases}
      \exists^{p\bmod q} x_{k+1}\colon\lambda(x_1,\dots,x_k,x_{k+1})\land\psi'
        &\text{ if }k=n-1\\
      \displaystyle
      \bigvee_{(*)} \bigwedge_{0\le i<q}
        \exists^{f(i)\bmod q}x_{k+1}\colon \alpha_i^{k+1}(x_1,\dots,x_k,x_{k+1})
         &\text{ otherwise}
    \end{cases}
  \]
  where the disjunction $(*)$ extends over all functions
  $f\colon\{0,1,\dots,q-1\}\to\{0,1,\dots,q-1\}$ with
  $\sum_{0\le i<q}i\cdot f(i)\equiv p\pmod q$.

  By induction, one obtains
  \[
    (\bN,+)\models\alpha_p^k(m_1,\dots,m_k)
  \]
  iff
  \[
    \biggl|
      \bigl\{(m_{k+1},m_{k+2},\dots,m_n)\mid (\bN,+) \models
                       \lambda(\overline{m})\land\psi'(\overline{m})
      \bigr\}
    \biggr|\in p+q\bN\,.
  \]
  Recall that $g$ maps the tuples satisfying $\lambda$ bijectively
  onto $L$. Hence, the above is equivalent to
  \[
    \biggl|
      \bigl\{w\in L \mid \exists m_{k+1},m_{k+2},\dots,m_n\in \bN\colon 
                          w=w_1^{m_1}\cdots w_n^{m_n}\in L\text{ and }
                          \cS\models\psi(w)
      \bigr\}
      %% \bigl\{(m_{k+1},m_{k+2},\dots,m_n)\mid 
      %%                     g(\overline{m})\in L\text{ and }
      %%                     \cS\models\psi(g(\overline{m}))
      %% \bigr\}
    \biggr|\in p+q\bN\,.
  \]
  Setting $\varphi'=\alpha_p^0$ therefore solves the problem.

  Consequently, any sentence $\varphi$ from $\FOMOD$ in the language
  of $\cS$ is translated into an equivalent sentence $\varphi'$ in the
  language of $(\bN,+)$. By \cite{Ape66,Sch05,HabK15}, validity of the
  sentence $\varphi'$ in $(\bN,+)$ is decidable.
\end{proof}

\section{The \CMOD$^2$-theory with regular predicates}

By \cite{KarS15}, the $\FO^2$-theory of
$(\Sigma^*,\subword,(L)_{L\text{ regular}})$ is decidable. This
two-variable fragment of first-order logic has a restricted expressive
power since, e.g., the following two properties cannot be expressed:
\begin{enumerate}
\item
  $x\cover y=(x\subword y\land x\neq y\land \forall z\colon(x\subword
  z\subword y\to (x=z\lor z=y))$.
\item
  $\exists x_1,x_2,x_3\colon\bigwedge_{1\le i<j\le 3} x_i\neq
  x_j\land\bigwedge_{1\le i\le 3}x_i\in L$.
\end{enumerate}
To make the first property accessible, we add the cover relation to
the structure. The logic $\CMOD^2$ allows to express the second
property with only two variables by $\exists^{\ge 3}x\colon x\in L$
(in addition, it can express that the regular language $L$ contains an
even number of elements which is not expressible in first-order logic
at all).

It is the aim of this section to show that the $\CMOD^2$-theory of the
structure
\[
  \cS=(\Sigma^*,\subword,\cover,(L)_{L\text{ regular}})
\]
is decidable. This decidability proof extends the proof from
\cite{KarS15} for the decidability of the $\FO^2$-theory of
$(\Sigma^*,\subword,(L)_{L\text{ regular}})$. That proof provides a
quantifier-elimination procedure that relies on two facts, namely
\begin{enumerate}
\item that the class of regular languages is closed under images under
  rational relations and
\item that the proper subword relation and the incomparablity relation are
  rational.
\end{enumerate}
Similarly, our more general result also provides a
quantifier-elimination procedure that relies on the following
extensions of these two properties:
\begin{enumerate}
\item The class of regular languages is closed under \emph{counting}
  images under \emph{unambiguous} rational relations
  (Section~\ref{SS-counting-closure-and-unambiguous-rational-relations})
  and
\item the proper subword, the cover, and the incomparability relation
  are \emph{unambiguous} rational
  (Section~\ref{SS-unambiguous-rational-relations}).
\end{enumerate}
The actual quantifier-elimination is then presented in
Section~\ref{SS-quantifier-elimination}.

\subsection{Unambiguous rational relations}
\label{SS-unambiguous-rational-relations}

Recall that, by Nivat's theorem~\cite{Ber79}, a relation
$R\subseteq\Sigma^*\times\Sigma^*$ is rational if there exist an
alphabet $\Gamma$, a homomorphism
$h\colon\Gamma^*\to\Sigma^*\times\Sigma^*$, and a regular language
$S\subseteq\Gamma^*$ such that $h$ maps $S$ surjectively onto $R$.  We
call $R$ \emph{unambiguous rational relation} if, in addition, $h$ maps
$S$ \emph{injectively} (and therefore bijectively) onto $R$.

\begin{example}
  The relations $R_1=\{(a^mba^n,a^m)\mid m,n\in\bN\}$ and
  $R_2=\{(a^mba^n,a^n)\mid m,n\in\bN\}$ are unambiguous rational: take
  $\Gamma=\{x,y,z\}$, $S=x^*yz^*$ and the homomorphisms
  \[
    x\mapsto (a,a)\ y\mapsto(b,\varepsilon)\ z\mapsto(a,\varepsilon)
  \]
  (for the first relation) and
  \[
    x\mapsto (a,\varepsilon)\ y\mapsto(b,\varepsilon)\ z\mapsto(a,a)\,.
  \]

  Note that the intersection $R_1\cap R_2$ is not even rational, while
  the union $R=R_1\cup R_2$ is rational (since the union of rational
  language is always rational) \cite{Ber79}. But this union is not
  unambiguous rational: If it were unambiguous rational, then the set
  \[
    \{u\in a^*ba^*\mid \exists^{\ge 2}v\colon(u,v)\in R\}
    =\{a^mba^n\mid m\neq n\}
  \]
  would be regular by Prop.~\ref{P-preservation-of-regularity} below.
\end{example}

\begin{lemma}\label{L-disjoint-union-of-unamb.-rational-relations}
  Let $R_1,R_2\subseteq\Sigma^*\times\Sigma^*$ be unambiguous rational
  and disjoint. Then $R_1\cup R_2$ is unambiguous rational.
\end{lemma}

\begin{proof}
  There are disjoint alphabets $\Gamma_1$ and $\Gamma_2$, regular
  languages $S_i\subseteq\Gamma_i$ and homomorphisms
  $h_i\colon\Gamma_i^*\to\Sigma^*\times\Sigma^*$ such that $S_i$ is
  mapped bijectively onto $R_i$. Let $\Gamma=\Gamma_1\cup\Gamma_2$ and
  $S=S_1\cup S_2$. Let the homomorphism
  $h\colon\Gamma^*\to\Sigma^*\times\Sigma^*$ be given by
  \[
    h(a)=
    \begin{cases}
      h_1(a) & \text{ if }a\in\Gamma_1\\
      h_2(a) & \text{ if }a\in\Gamma_2
    \end{cases}
  \]
  for $a\in\Gamma$. Then $h$ maps the regular language~$S$ bijectively
  onto $R_1\cup R_2$.
\end{proof}

\begin{lemma}\label{L-subword-unamb.-rational}
  For any alphabet $\Sigma$, the cover relation $\cover$ and the
  relation $\propersubword\setminus\cover$ are unambiguous rational.
\end{lemma}

\begin{proof}
  For $i\in\{1,2\}$, let $\Sigma_i=\Sigma\times\{i\}$ and
  $\Gamma=\Sigma_1\cup\Sigma_2$. Furthermore, let the homomorphism
  $\proj_i\colon\Gamma^*\to\Sigma^*$ be defined by $\proj_i(a,i)=a$
  and $\proj_i(a,3-i)=\varepsilon$ for all $a\in\Sigma$.  Finally, let
  the homomorphism $\proj\colon\Gamma^*\to\Sigma^*\times\Sigma^*$ be
  defined by $\proj(w)=(\proj_1(w),\proj_2(w))$.

  Now consider the regular language
  \[
    \Sub = \left(\bigcup_{a\in\Sigma}\Bigl(\bigl(\Sigma_2\setminus\{(a,2)\}\bigr)^*\, (a,2)\,(a,1)\Bigr)\right)^*
    {\Sigma_2}^*\,.
  \]

  Let $w\in \Sub$. Since any occurrence of a letter $(a,1)$ in $w$ is
  immediately preceeded by an occurrence of $(a,2)$, we get
  $\proj_1(w)\subword\proj_2(w)$.

  Conversely, if $u=a_1a_2\dots a_n\subword v$, then $v$ can be
  written (uniquely) as $v=v_1a_1\,v_2a_2\cdots v_na_n\,v_{n+1}$ such
  that $a_i$ does not occur in $v_i$ (for all
  $i\in\{1,2,\dots,n\}$). For $i\in[1,n+1]$ let $w_i\in\Gamma^*$ be
  the unique word from $\Sigma_2^*$ with $\proj_2(w_i)=v_i$. Then
  \[
    w=w_1(a_1,2)(a_1,1)\,w_2(a_2,2)(a_2,1)\,w_3(a_3,2)(a_3,1)\cdots w_n(a_n,2)(a_n,1)\,w_{n+1}
  \]
  belongs to $\Sub$ and satisfies $\proj(w)=(u,v)$.

  Since the factorization of $v$ is unique, we have that $\proj$ maps
  $\Sub$ bijectively onto the subword relation $\subword$.
  \begin{itemize}
  \item Let $S$ denote the intersection of $\Sub$ with
    $\bigl(\Sigma_2\Sigma_1\bigr)^*\,\Sigma_2\,\bigl(\Sigma_2\Sigma_1\bigr)^*$,
    i.e., the regular language of words from $\Sub$ with precisely one
    more occurrence of letters from $\Sigma_2$ than from
    $\Sigma_1$. Then $S$ is mapped bijectively onto the relation
    $\cover$, hence this relation is unambiguous rational. 
  \item Similarly, let $S'$ denote the regular language of all words
    from $\Sub$ with at least two more occurrences of letters from
    $\Sigma_2$ than from $\Sigma_1$. It is mapped bijectively onto the
    relation $\propersubword\setminus\cover$. Hence this relation is
    unambiguous rational.
  \end{itemize}
\end{proof}

\begin{lemma}\label{L-inc-unamb.-rational}
  For any alphabet $\Sigma$, the incomparability relation
  \[
    \mathbin{\parallel}=\{(u,v)\in\Sigma^*\times\Sigma^*\mid
    \text{neither }u\subword v\text{ nor }v\subword u\}
  \]
  is unambiguous rational.
\end{lemma}

\begin{proof}
  Note that the set $\parallel$ is the disjoint union of the following
  three relations:
  \begin{enumerate}
  \item $R_1=\{(u,v)\mid |u|<|v|\text{ and not }u\subword v\}$,
  \item $R_2=\{(u,v)\mid |u|=|v|\text{ and  }u\neq v\}$, and
  \item $R_3=\{(u,v)\mid |u|>|v|\text{ and not }v\subword u\}$.
  \end{enumerate}

  As in the previous proof, let $\Sigma_i=\Sigma\times\{i\}$ and
  $\Gamma=\Sigma_1\cup\Sigma_2$. Furthermore, let the homomorphism
  $\proj_i\colon\Gamma^*\to\Sigma^*$ be defined by $\proj_i(a,i)=a$
  and $\proj_i(a,3-i)=\varepsilon$ for all $a\in\Sigma$.
  
  We prove that the relations $R_1$, $R_2$, and $R_3$ are all
  unambiguous rational. From
  Lemma~\ref{L-disjoint-union-of-unamb.-rational-relations}, we then
  get that $R$ is unambiguous rational since it is the disjoint union
  of these three relations.

  We start with the simple case: $R_2$. Consider the regular language
  \[
    \Inc_2=(\Sigma_1\Sigma_2)^*\cdot
      \{(a,1)(b,2)\mid a,b\in\Sigma, a\neq b\}
      \cdot (\Sigma_1\Sigma_2)^*\,.
  \]
  This is the set of sequences of words of the form $(a,1)(b,2)$ such
  that, at least once, $a\neq b$. Hence, $\proj$ maps the regular
  language $\Inc_2$ bijectively onto $R_2$.
  
  Next, we handle the relation $R_1\cup R_2$. Correcting \cite[Lemma 5.2]{KarS15}
  slightly, we learn that $(u,v)\in R_1\cup R_2$ if, and only if,
  \begin{itemize}
  \item $u=a_1a_2\dots a_\ell u'$ for some $\ell\ge1$,
    $a_1,\dots,a_\ell\in\Sigma$,  $u'\in\Sigma^*$, and
  \item
    $v\in
    (\Sigma\setminus\{a_1\})^*a_1\,(\Sigma\setminus\{a_2\})^*a_2\cdots
    (\Sigma\setminus\{a_{\ell-1}\})^*a_{\ell-1}\,
    (\Sigma\setminus\{a_\ell\})^+ v'$ for some word $v'\in\Sigma^*$
    with $|u'|=|v'|$.
  \end{itemize}
  Furthermore, the number $\ell$, the letters $a_1,a_2,\dots,a_\ell$,
  and the words $u'$ and $v'$ are unique. Define
  \[
    \Inc_{1,2}=
    \left(
      \bigcup_{a\in\Sigma}
        \Bigl(
          \bigl(\Sigma_1\setminus\{(a,1)\}\bigr)^*(a,1)(a,2)
        \Bigr)
    \right)^*
    \cdot
    \bigcup_{a\in\Sigma}
      \Bigl(\bigl(\Sigma_1\setminus\{(a,1)\}\bigr)^+(a,2)\Bigr)
    \cdot
    \left(\Sigma_1\Sigma_2\right)^*\,.
  \]
  By the above characterisation of $R_1\cup R_2$, the homomorphism
  $\proj$ maps $\Inc_{1,2}$ bijectively onto $R_1\cup R_2$. 

  Recall that $\Inc_2\subseteq\Inc_{1,2}$ is mapped bijectively onto
  $R_2$. Hence $\proj$ maps $\Inc_1=\Inc_{1,2}\setminus\Inc_2$
  bijectively onto $R_1$.

  Since the class of unambiguous rational relations is closed under
  inverses, also $R_3=R_1^{-1}$ is unambiguous rational.
\end{proof}

\subsection{Closure properties of the class of regular languages}
\label{SS-counting-closure-and-unambiguous-rational-relations}

Let $R\subseteq\Sigma^*\times\Sigma^*$ be an unambiguous rational
relation and $L\subseteq\Sigma^*$ a regular language. We want to show
that the languages of all words $u\in\Sigma^*$ with
\begin{align}
  & \text{ with }|\{v\in L\mid (u,v)\in R\}|\ge k\label{threshold-language}\\
  & (\text{with }|\{v\in L\mid (u,v)\in R\}|\in p+q\bN\text{, respectively})
    \label{modulo-language}
\end{align}
are effectively regular for all $k\in\bN$ and all $0\le p<q$,
respectively.

\begin{example}
  Consider the rational relation
  $R=\{(a^kba^\ell,a^m)\mid k=m\text{ or }\ell=m\}$ and the regular
  language $L=\Sigma^*$. With $k=2$, the language
  \eqref{threshold-language} equals the non-regular set
  $\{a^kba^\ell\mid k\neq \ell\}$. Thus, to prove effective
  regularity, we need to restrict the rational relation~$R$.
\end{example}

For these proofs, we need the following classical concepts.  Let $S$
be a semiring. A function $r\colon\Sigma^*\to S$ is \emph{realizable
  over $S$}, if there are $n\in\bN$, $\lambda\in S^{1\times n}$, a
homomorphism $\mu\colon\Sigma^*\to S^{n\times n}$, and
$\nu\in S^{n\times 1}$ with $r(w)=\lambda\cdot \mu(w)\cdot\nu$ for all
$w\in\Sigma^*$.\footnote{In the literature, a realizable function is
  often called recognizable formal power series. Since, in this paper,
  we will not encounter any operations on formal power series (like
  addition, Cauchy product etc), we use the (in this context) more
  intuitive notion of a ``realizable function''.} The triple
$(\lambda,\mu,\nu)$ is a \emph{presentation} or a \emph{weighted
  automaton for $r$}.

In the following, we consider the semiring $\bN^\infty$, i.e., the set
$\bN\cup\{\infty\}$ together with the commutative operations $+$ and
$\cdot$ (with $x+\infty=\infty$ for all $x\in\bN\cup\{\infty\}$,
$x\cdot\infty=\infty$ for all $x\in(\bN\cup\{\infty\})\setminus\{0\}$,
and $0\cdot\infty=0$). On this set, we define (in a natural way) an
infinite sum setting
\[
  \sum_{i\in I}x_i=
  \begin{cases}
    \infty  & \text{ if there are infinitely many $i\in I$ with $x_i>0$}\\
    \displaystyle
    \sum_{\substack{i\in I\\x_i>0}}x_i & \text{ otherwise}.
  \end{cases}
\]
for any family $(x_i)_{i\in I}$ with entries in $\bN^\infty$.

Our first aim in this section is to prove the following

\begin{proposition}\label{P-general-preservation}
  Let $\Gamma$ and $\Sigma$ be alphabets, $f\colon\Gamma^*\to\Sigma^*$
  a homomorphism, and $\chi\colon\Gamma^*\to\bN^\infty$ a realizable
  function over $\bN^\infty$. Then the function
  \[
    r=\chi\circ f^{-1}\colon\Sigma^*\to\bN^\infty
    \colon u\mapsto\sum_{\substack{w\in\Gamma^*\\f(w)=u}}\chi(w)
  \]
  is effectively realizable over $\bN^\infty$.
\end{proposition}

Before we can do this in all generality, we first consider two special
cases: A monoid homomorphism $f\colon\Gamma^*\to\Sigma^*$ between free
monoids is \emph{non-expanding} if $|f(w)|\le |w|$ for all
$w\in\Gamma^*$, i.e. $f(a)\in\Sigma\cup\{\varepsilon\}$ for all
$a\in\Gamma$. It is \emph{non-erasing} if, dually, $|f(w)|\ge|w|$ for
all $w\in\Gamma^*$, i.e., $f(a)\in\Sigma^+$ for all $a\in\Gamma$.

\begin{lemma}\label{L-non-expanding-preservation}
  Let $\Gamma$ and $\Sigma$ be alphabets, $f\colon\Gamma^*\to\Sigma^*$
  a non-expanding homomorphism, and $\chi\colon\Gamma^*\to\bN^\infty$
  a realizable function over $\bN^\infty$. Then the function
  \[
    r=\chi\circ f^{-1}\colon\Sigma^*\to\bN^\infty
    \colon u\mapsto\sum_{\substack{w\in\Gamma^*\\f(w)=u}}\chi(w)
  \]
  is effectively realizable over $\bN^\infty$.
\end{lemma}

\begin{proof}
  Let $(\lambda,\mu,\nu)$ be a presentation of dimension $n$ for
  $\chi$. 

  For $\sigma\in\Sigma\cup\{\varepsilon\}$, let
  \[
    \Gamma_\sigma=\{b\in\Gamma\mid f(b)=\sigma\}\,.
  \]
  Since $f$ is non-expanding, $\Gamma$ is the disjoint union of these
  subalphabets. Furthermore, let $M\in(\bN^\infty)^{n\times n}$ be the
  matrix defined by
  \begin{equation}
    M_{ij}=\sum_{w\in\Gamma_\varepsilon^*}\mu(w)_{ij}
    \label{eq:matrix-M}
  \end{equation}
  for all $i,j\in[1,n]$.

  To define a presentation for the function $r$, we first
  define a homomorphism
  $\mu'\colon\Sigma^*\to\bigl(\bN^\infty\bigr)^{n\times n}$ by
  \[
    \mu'(a)=\sum_{b\in\Gamma_a} \bigl(\mu(b)\cdot M\bigr)
  \]
  for all $a\in\Sigma$. Setting
  \[
    \lambda'=\lambda\cdot M \text{ and } \nu' =\nu
  \]
  defines the presentation $(\lambda',\mu',\nu')$ of dimension
  $n$. Now let $u=a_1a_2\dots a_m\in\Sigma^*$ with $a_i\in\Sigma$ for
  all $1\le i\le m$. Then we get
  \begin{align*}
    \lambda'\cdot\mu'(u)\cdot\nu'
    &=\lambda\cdot M
      \cdot
      \left(
        \prod_{1\le i\le m}
          \sum_{b_i\in \Gamma_{a_i}}\bigl(\mu(b_i)\cdot M\bigr) 
      \right)
      \cdot\nu\\
    &=\lambda\cdot\sum\bigl(\mu(w_0)\mid w_0\in\Gamma_\varepsilon^*\bigr)
      \cdot
      \left(
        \prod_{1\le i\le m}
          \sum\bigl(\mu(w_i)\mid w_i\in\Gamma_{a_i}\Gamma_\varepsilon^*\bigr)
      \right)
      \cdot\nu\\
    &=\lambda\cdot
      \sum\bigl(\mu(w)\mid w\in\Gamma_\varepsilon^*\Gamma_{a_1}\Gamma_\varepsilon^*\Gamma_{a_2}\cdots\Gamma_\varepsilon^*\Gamma_{a_m}\Gamma_\varepsilon^*
      \bigr)
      \cdot\nu\\
    &=\lambda\cdot
      \sum\bigl(\mu(w)\mid w\in\Gamma^*, f(w)=u\bigr)
      \cdot\nu\\
    &= r(u)\,.
  \end{align*}
  Hence, $(\lambda',\mu',\nu')$ is a presentation for the function
  $r$, i.e., $r$ is realizable.

  It remains to be shown that the presentation $(\lambda',\mu',\nu')$
  is computable from the presentation $(\lambda,\mu,\nu)$ and the
  homomorphism $f$. For this, it suffices to construct the matrix $M$
  effectively, i.e., to compute the infinite sum in
  Eq.~\eqref{eq:matrix-M}. Using a pumping argument, one first shows
  the equivalence of the following two statements for all
  $i,j\in\{1,2,\dots,n\}$:
  \begin{enumerate}[(a)]
  \item There are infinitely many words $w\in\Gamma_\varepsilon^*$
    with $\mu(w)_{ij}>0$.
  \item There is a word $w\in\Gamma_\varepsilon^*$ with $n<|w|\le 2n$
    and $\mu(w)_{ij}>0$.
  \end{enumerate}
  Since statement (b) is decidable, we can evaluate
  Eq.~\eqref{eq:matrix-M} calculating
  \[
    M_{ij}=
    \begin{cases}
      \infty & \text{ if (b) holds}\\
      \sum \bigl(\mu(w)_{ij}\mid w\in\Gamma_\varepsilon^*, |w|\le n\bigl)
             & \text{ otherwise.}
    \end{cases}
  \]           
\end{proof}

\begin{lemma}\label{L-non-erasing-preservation}
  Let $\Gamma$ and $\Sigma$ be alphabets, $f\colon\Gamma^*\to\Sigma^*$
  a non-erasing homomorphism, and $\chi\colon\Gamma^*\to\bN^\infty$ a
  realizable function over $\bN^\infty$. Then the function
  \[
    r=\chi\circ f^{-1}\colon\Sigma^*\to\bN^\infty
    \colon u\mapsto\sum_{\substack{w\in\Gamma^*\\f(w)=u}}\chi(w)
  \]
  is effectively realizable over $\bN^\infty$.
\end{lemma}

\begin{proof}
  Since $\chi$ is realizable, it can be constructed from functions
  $s\colon\Gamma^*\to\bN^\infty$ with $s(w)\neq0$ for at most one
  $w\in\Gamma^*$ using addition, Cauchy-product, and iteration applied
  to functions $t$ with $t(\varepsilon)=0$
  \cite[Theorem~3.11]{Sak09}. Replacing, in this construction, the
  basic function $s$ by $s'\colon\Sigma^*\to\bN^\infty$ with
  \[
    f'(x)=\sum_{\substack{w\in\Gamma^*\\f(w)=x}}s(w)
  \]
  yields a construction of $r$. Since $f$ is non-erasing, also in this
  construction, iteration is only applied to functions $t$ with
  $t(\varepsilon)=0$. Hence, by \cite[Theorem~3.1]{Sak09} again, $r$
  is realizable. Analysing the proof of that theorem, one even obtains
  that a presentation for $r$ can be computed from $f$ and a
  presentation of $\chi$.
\end{proof}

% \begin{proposition}\label{P-general-preservation}
%   Let $\Gamma$ and $\Sigma$ be alphabets, $f\colon\Gamma^*\to\Sigma^*$
%   a homomorphism, and $\chi\colon\Gamma^*\to\bN^\infty$ a realizable
%   function over $\bN^\infty$. Then the function
%   \[
%     r=\chi\circ f^{-1}\colon\Sigma^*\to\bN^\infty
%     \colon u\mapsto\sum_{\substack{w\in\Gamma^*\\f(w)=u}}\chi(w)
%   \]
%   is effectively realizable over $\bN^\infty$.
% \end{proposition}

\begin{proof}[Proof of Prop.~\ref{P-general-preservation}]
  Let $\sigma\in\Sigma$ be arbitrary. Then define homomorphisms
  $f_1\colon \Gamma^*\to\Gamma^*$ and $f_2\colon\Gamma^*\to\Sigma^*$
  by
  \[
    f_1(a)=
    \begin{cases}
      \varepsilon &\text{if }f(a)=\varepsilon\\
      a & \text{otherwise}
    \end{cases}
    \text{ and }
    f_2(a)=
    \begin{cases}
      f(a) &\text{if }f(a)\neq\varepsilon\\
      \sigma & \text{otherwise}
    \end{cases}
  \]
  for all letters $a\in\Gamma$. Then $f_1$ is non-expanding, $f_2$ is
  non-erasing, and $f=f_2\circ f_1$.

  By Lemma~\ref{L-non-expanding-preservation}, the function
  $\chi\circ f_1^{-1}$ is effectively realizable. By
  Lemma~\ref{L-non-erasing-preservation}, also
  $\chi\circ f_1^{-1}\circ f_2^{-1}=\chi\circ f^{-1}$ is effectively
  realizable.  
\end{proof}

\begin{lemma}\label{L-unamb.-rat.-to-real.}
  Let $R\subseteq\Sigma^*\times\Sigma^*$ be an unambiguous rational
  relation and $L\subseteq\Sigma^*$ be regular. Then the function
  \[
    r\colon \Sigma^*\to\bN^\infty\colon u\mapsto|\{v\in L\mid (u,v)\in R\}|
  \]
  is effectively realizable.
\end{lemma}

\begin{proof}
  Since $R$ is unambiguous rational, there are an alphabet $\Gamma$,
  homomorphisms $f,g\colon\Gamma^*\to \Sigma^*$, and a regular
  language $S\subseteq\Gamma^*$ such that
  \[
    (f,g)\colon\Gamma^*\to\Sigma^*\times\Sigma^*
         \colon w\mapsto\bigl(f(w),g(w)\bigr)
  \]
  maps $S$ bijectively onto the relation $R$. Let
  $S_L=S\cap g^{-1}(L)$. Since $L$ is regular, the language $S_L$ is
  effectively regular. Furthermore, $(f,g)$ maps $S_L$ bijectively
  onto $R\cap (\Sigma^*\times L)$.

  Since $S_L$ is regular, the characteristic function
  \[
    \chi\colon\Gamma^*\to \bN^\infty\colon w\mapsto
    \begin{cases}
      1 & \text{ if } w\in S_L\\
      0 & \text{ otherwise}
    \end{cases}
  \]
  is effectively realizable over $\bN^\infty$ \cite[Proposition~3.12]{Sak09}.

  By Proposition~\ref{P-general-preservation}, also the function
  \[
    r'\colon\Sigma^*\to \bN^\infty
      \colon u\mapsto \sum_{\substack{w\in \Gamma^*\\f(w)=u}}\chi(w)
  \]
  is effectively realizable over $\bN^\infty$. Note that, for $u\in\Sigma^*$,
  we get
  \begin{align*}
    r'(u)
    &= \sum_{\substack{w\in\Gamma^*\\f(w)=u}}\chi(w)\\
    &= |\{w\in S_L\mid f(w)=u\}|
      &&\text{ since $\chi$ is the characteristic function of }S_L\\
    &= |\{g(w)\mid  w\in S_L, f(w)=u\}|
      &&\text{ since }(f,g)\text{ is injective on }S_L\subseteq S\\
    &= |\{v\mid (u,v)\in R\cap(\Sigma^*\times L)\}|
      &&\text{ since }(f,g)\text{ maps $S_L$ onto }R\cap(\Sigma^*\times L)\\
    &= |\{v\in L\mid (u,v)\in R\}|\,.
  \end{align*}
  In other words, the effectively realizable function $r'$ is the
  function $r$ from the statement of the lemma.
\end{proof}

\begin{proposition}\label{P-preservation-of-regularity}
  Let $R\subseteq\Sigma^*\times\Sigma^*$ be an unambiguous rational
  relation and $L\subseteq\Sigma^*$ be regular.
  \begin{enumerate}
  \item For $k\in\bN$, the set of words $u\in\Sigma^*$ with
    \[
      |\{v\in L\mid (u,v)\in R\}|\ge k
    \]
    is effectively regular.
  \item For $p,q\in\bN$ with $p<q$, the set $H$ of words
    $u\in\Sigma^*$ with
    \[
      |\{v\in L\mid (u,v)\in R\}|\in p+q\bN
    \]
    is effectively regular.
  \end{enumerate}
\end{proposition}

\begin{proof}
  By Lemma~\ref{L-unamb.-rat.-to-real.}, the function
  $r\colon\Sigma^*\to\bN^\infty\colon u\mapsto|\{v\in L\mid (u,v)\in
  R\}|$ is effectively realizable over $\bN^\infty$.

  We construct a semiring
  $S_k^\infty=(\{0,1,\dots,k,\infty\},\oplus,\odot,0,1)$ setting
  \[
    x\oplus y=
    \begin{cases}
      x+y & \text{ if }x+y\in S_k^\infty\\
      k &\text{ otherwise}
    \end{cases}
    \text{ and } 
    x\odot y=
    \begin{cases}
      x\cdot y & \text{ if }x\cdot y\in S_k^\infty\\
      k &\text{ otherwise}
    \end{cases}
  \]
  for $x,y\in S_k^\infty$ (note that $x+y\notin S_k^\infty$ is
  equivalent to $k<x+y<\infty$).

  Then the mapping
  \[
    h\colon\bN^\infty\to S_k^\infty\colon n\mapsto
    \begin{cases}
      \min\{k,n\} & \text{ if }n\in\bN\\
      \infty & \text{ if }n=\infty
    \end{cases}
  \]
  is a semiring homomorphism. It follows that the function
  \[
    h\circ r\colon\Sigma^*\to S_k^\infty\colon
    u\mapsto\min\{k,|\{u\in L\mid (u,v)\in R\}|\}
  \]
  is effectively realizable over $S_k^\infty$
  \cite[Prop.~4.5]{Sak09}. Since the semiring $S_k^\infty$ is finite,
  the language
  \[
    (h\circ r)^{-1}(k)=\{u\in\Sigma^*\mid r(u)\ge k\}
  \]
  is effectively regular \cite[Prop.~6.3]{Sak09}. Since
  $r(u)\ge k\iff |\{v\in L\mid (u,v)\in R\}|\ge k$, this proves the
  first statement.

  The second statement is shown similarly. We construct the semiring
  $\bZ_q=(\{0,1,\dots,q,\infty\},\oplus,\odot,0,1)$ with
  \[
    x\oplus y=
    \begin{cases}
      x+y & \text{ if }x+y\in \bZ_q^\infty\\
      x+y \bmod q & \text{ otherwise}
    \end{cases}
    \text{ and } 
    x\odot y=
    \begin{cases}
      x\cdot y & \text{ if }x\cdot y\in\bZ_q^\infty\\
      x\cdot y\bmod q &\text{ otherwise}
    \end{cases}
  \]
  for $x,y\in S_k^\infty$. Note that, with $q=6$, we get
  $4\oplus 2=6\neq 0=(4+2)\bmod q$ and similarly
  $3\odot 2=6\neq 0=(3\cdot2)\bmod q$.\footnote{The reader might
    wonder why both, $0$ and $q$, belong to $\bZ_q^\infty$. Suppose we
    identified them, i.e., considered $S=\{0,1,\dots,q-1,\infty\}$. Then
    we would get
    $0 = 0\odot \infty = (1\oplus(q-1))\odot\infty =
    1\odot\infty\oplus(q-1)\odot\infty=\infty\oplus\infty=\infty$.}

  Let $\eta$ denote the semiring homomorphism from $\bN^\infty$ to
  $\bZ_q^\infty$ with
  \[
    \eta(n)=
    \begin{cases}
      \infty & \text{ if }n=\infty\\
      q & \text{ if }n\in q\bN\setminus\{0\}\\
      n\bmod q & \text{ otherwise.}
    \end{cases}
  \]
  It follows that the function
  \[
    \eta\circ r\colon\Sigma^*\to \bZ_q^\infty\colon
        u\mapsto \eta(|\{v\in L\mid (u,v)\in R\}|)
  \]
  is effectively realizable over $\bZ_q^\infty$
  \cite[Prop.~4.5]{Sak09}. Since the semiring $\bZ_q^\infty$ is
  finite, the language
  \[
    (\eta\circ r)^{-1}(x)
  \]
  is effectively regular for all $x\in \bZ_q^\infty$
  \cite[Prop.~6.3]{Sak09}. Then the claim follows since
  \[
    H =
    \begin{cases}
      (\eta\circ r)^{-1}(p) & \text{ if }1\le p<q\\
      (\eta\circ r)^{-1}(0)\cup(\eta\circ r)^{-1}(q)
                              & \text{ if }p=0\,.
    \end{cases}
  \]
\end{proof}

\subsection{Quantifier elimination for $\CMOD^2$}
\label{SS-quantifier-elimination}

Our decision procedure employs a quantifier alternation procedure,
i.e., we will transform an arbitrary formula into an equivalent one
that is quantifer-free. As usual, the heart of this procedure handles
formulas $\psi=\mathrm{Q}y\,\varphi$ where $\mathrm{Q}$ is a
quantifier and $\varphi$ is quantifier-free. Since the logic $\CMOD^2$
has only two variables, any such formula $\psi$ has at most one free
variable. In other words, it defines a language $K$.  The following
lemma shows that this language is effectively regular, such that
$\psi$ is equivalent to the quantifier-free formula $x\in K$.

\begin{lemma}\label{L-quantifier-elimination}
  Let $\varphi(x,y)$ be a quantifier-free formula from $\CMOD^2$. Then
  the sets
  \[
    \{x\in\Sigma^*\mid \cS\models\exists^{\ge k}y\,\varphi\}\text{ and
    } \{x\in\Sigma^*\mid \cS\models\exists^{p\bmod q}y\,\varphi\}
  \]
  are effectively regular for all $k\in\bN$ and all $p,q\in\bN$ with $p<q$.
\end{lemma}

\begin{proof}
  Without changing the meaning of the formula $\varphi$, we can do the
  following replacements of atomic formulas:
  \begin{itemize}
  \item $x=y$ can be replaced by $x\subword y\land y\subword x$,
  \item $x\subword x$ and $y\subword y$ by $x\in\Sigma^*$, and
  \item $x\cover x$ and $y\cover y$ by $x\in\emptyset$.
  \end{itemize}
  Since $\varphi$ is quantifier-free, we can therefore assume that it
  is a Boolean combination of formulas of the form
  \begin{itemize}
  \item $x\in K$ for some regular language $K$,
  \item $y\in L$ for some regular language $L$,
  \item $x\subword y$, 
  \item $y\subword x$,
  \item $x\cover y$, and 
  \item $y\cover x$.
  \end{itemize}

  We define the following formulas $\theta_i(x,y)$ for $1\le i\le 6$:
  \[
    \theta_i(x,y)=
    \begin{cases}
      x=y & \text{ if }i=1\\
      x\cover y & \text{ if }i=2\\
      x\subword y\land\lnot(x=y\lor x\cover y) & \text{ if }i=3\\
      y\cover x & \text{ if }i=4\\
      y\subword x\land\lnot(y=x\lor y\cover x) & \text{ if }i=5\\
      \lnot(x\subword y\lor y\subword x) & \text{ if }i=6
    \end{cases}
  \]
  Note that any pair of words $x$ and $y$ satisfies precisely one of
  these six formulas. Hence $\varphi$ is equivalent to
  \[
    \bigvee_{1\le i\le 6}\bigl(\theta_i\land\varphi)\,.
  \]
  In this formula, any occurrence of $\varphi$ appears in conjunction
  with precisely one of the formulas $\theta_i$. Depending on this
  formula $\theta_i$, we can simplify $\varphi$ to $\varphi_i$ by
  replacing the atomic subformulas that compare $x$ and $y$ as
  follows:
  \begin{itemize}
  \item If $i\in\{1,2,3\}$, we replace $x\subword y$ by the valid
    formula $\top=(x\in\Sigma^*)$.
  \item If $i\in\{1,4,5\}$, we replace $y\subword x$ by $\top$.
  \item If $i=2$, we replace $x\cover y$ by $\top$.
  \item If $i=4$, we replace $y\cover x$ by $\top$.
  \end{itemize}
  All remaining comparisions are replaced by $\bot=(x\in\emptyset)$.

  As a result, the formula $\varphi$ is equivalent to
  \[
    \bigvee_{1\le i\le 6}\bigl(\theta_i\land\varphi_i)
  \]
  where the formulas $\varphi_i$ are Boolean combinations of formulas
  of the form $x\in K$ and $y\in L$ for some regular languages $K$ and
  $L$.
  
  Now let $k\in\bN$. Since the formulas $\theta_i$ are mutually
  exclusive (i.e., $\theta_i(x,y)\land\theta_j(x,y)$ is satisfiable
  iff $i=j$), we get
  \[
    \exists^{\ge k}y\,\varphi
    \equiv
    \exists^{\ge k}y\,\bigvee_{1\le i\le 6}(\theta_i\land\varphi_i)
    \equiv
      \bigvee_{(*)}
       \bigwedge_{1\le i\le 6}
          \exists^{\ge k_i}y\,(\theta_i\land\varphi_i)
  \]
  where the disjunction $(*)$ extends over all tuples
  $(k_1,\dots,k_6)$ of natural numbers with $\sum_{1\le i\le 6}k_i=k$.

  Hence it suffices to show that
  \begin{equation}
    \label{eq:formula-theta}
    \{x\in\Sigma^*\mid \exists^{\ge k}y\,(\theta_i \land\varphi)\}
  \end{equation}
  is effectively regular for all $1\le i\le 6$, all $k\in\bN$, and all
  Boolean combinations $\varphi$ of formulas of the form $x\in K$ and
  $y\in L$ where $K$ and $L$ are regular languages. Since the class of
  regular languages is closed under Boolean operations, we can find
  regular languages $K_i$ and $L_i$ such that $\varphi$ is equivalent
  to
  \[
    \bigvee_{1\le i\le n}(x\in K_i\land y\in L_i)\,.
  \]

  Note that this formula is equivalent to
  \[
    \bigvee_{M\subseteq\{1,\dots,n\}}
    \left(
      x\in \underbrace{\bigcap_{i\in M}K_i\setminus\bigcup_{i\notin M}K_i}
           _{=K_M}
      \land
      y\in \underbrace{\bigcup_{i\in M}L_i}
           _{=L_M}
    \right)\,.
  \]
  Since this disjunction is exclusive (i.e. any pair of words $(x,y)$
  satisfies at most one of the cases), the set from
  \eqref{eq:formula-theta} equals the union of the sets
  \begin{equation}
    \label{eq:formula-theta2}
    \{x\in\Sigma^*\mid \exists^{\ge k}y\,(\theta_i \land x\in K_M\land y\in L_M)\}
  \end{equation}
  for $M\subseteq\{1,2,\dots,n\}$. Observe that for $k=0$, this set equals $\Sigma^*$ and we are done. So let us assume $k\ge 1$ from now on. Note that in that case, the set from
  \eqref{eq:formula-theta2} equals
  \[
      K_M\cap\{x\in\Sigma^*\mid \exists^{\ge k}y\in L_M\colon\theta_i\}\,.
  \]
  This set, in turn, equals
  \begin{itemize}
  \item $K_M\cap L_M$ if $i=1$ and $k=1$,
  \item $\emptyset$ if $i=1$ and $k>1$,
  \item $K_M\cap \{x\in\Sigma^*\mid\exists^{\ge k}y\in L_M\colon x\cover y\}$
    if $i=2$,
  \item
    $K_M\cap \{x\in\Sigma^*\mid\exists^{\ge k}y\in L_M\colon (x,
    y)\in{\propersubword\setminus\cover}\}$ if $i=3$,
  \item $K_M\cap \{x\in\Sigma^*\mid\exists^{\ge k}y\in L_M\colon y\cover x\}$
    if $i=4$,
  \item
    $K_M\cap \{x\in\Sigma^*\mid\exists^{\ge k}y\in L_M\colon (y,
    x)\in{\propersubword\setminus\cover}\}$ if $i=5$, and
  \item
    $K_M\cap \{x\in\Sigma^*\mid\exists^{\ge k}y\in L_M\colon
    x\parallel y\}$ if $i=6$.
  \end{itemize}
  In any case, it is effectively regular by
  Prop.~\ref{P-preservation-of-regularity},
  Lemma~\ref{L-subword-unamb.-rational}, and
  Lemma~\ref{L-inc-unamb.-rational}. Since the language from the claim
  of the lemma is a Boolean combination of such languages, the first
  claim is demonstrated.

  To also demonstrate the regularity of the second language, let
  $p,q\in\bN$ with $p<q$. Then $\exists^{p\bmod q}y\,\varphi$ is
  equivalent to the disjunction of all formulas of the form
  \[
     \bigwedge_{1\le i\le 6}
        \exists^{p_i\bmod q}y\,(\theta_i\land\varphi_i)
  \]
  where $(p_1,\dots,p_6)$ is a tuple of natural numbers from
  $\{0,1,\dots,q-1\}$ with $\sum_{1\le i\le 6}p_i\equiv p\bmod q$.
  The rest of the proof proceeds \textit{mutatis mutandis}.
\end{proof}

\begin{theorem}
  Let $\cS=(\Sigma^*,\subword,\cover,(L)_{L\text{ regular}})$. Let
  $\varphi(x)$ be a formula from $\CMOD^2$. Then the set
  \[
    \{x\in\Sigma^*\mid \cS\models\varphi(x)\}
  \]
  is effectively regular.
\end{theorem}

\begin{proof}
  The claim is trivial if $\varphi$ is atomic. For more complicated
  formulas, the proof proceeds by induction using
  Lemma~\ref{L-quantifier-elimination} and the effective closure of
  the class of regular languages under Boolean operations.
\end{proof}

\begin{corollary}\label{C-CMOD2}
  Let $L\subseteq\Sigma^*$ be a regular language.  Then the
  $\CMOD^2$-theory of the structure $\mathcal
  S=(L,\subword,\cover,(K\cap L)_{K\text{ regular}},(w)_{w\in L})$ is
  decidable.
\end{corollary}

\begin{proof}
  Let $\varphi\in\CMOD^2$ be a sentence. By the previous theorem, the set
  \[
    \{x\in L\mid\cS\models\varphi\}
  \]
  is regular. Hence $\varphi$ holds iff this set is nonempty, which is
  decidable.
\end{proof}

\section{The $\Sigma_1$-theory}

Let $L$ be regular and bounded. Then, by Theorem~\ref{T-bounded}, we
obtain in particular that the $\Sigma_2$-theory of $(L,\subword)$ is
decidable. Note that the regular language $L=\{a,b\}^*$ is not
covered by this result since it is unbounded. And, indeed, the
$\Sigma_2$-theory of $(\{a,b\}^*,\subword)$ is
undecidable~\cite{HalSZ17}.  On the positive side, we know that the
$\Sigma_1$-theory of $(\{a,b\}^*,\subword)$ is
decidable~\cite{Kus06}.

In this section, we generalize this positive result to arbitrary
regular languages, i.e., we prove the following result:

\begin{theorem}\label{T-unbounded}
  Let $L\subseteq\Sigma^*$ be regular. Then the $\Sigma_1$-theory of
  $\cS=(L,\subword)$ is decidable.
\end{theorem}

The proof for the case $L=\{a,b\}^*$ in~\cite{Kus06} essentially
relies on the fact that each order $(\bN^k,\le)$, and thus every
finite partial order, embeds into $(\{a,b\}^*,\subword)$.

In the general case here, the situation is more involved.  Take, for
example, $L=\{ab,ba\}^*$. Then, orders as simple as $(\bN^2,\le)$ do
not embed into $(L,\subword)$: This is because the downward closure of
any infinite subset of $L$ contains all of $L$, but $\bN^2$ contains a
downwards closed infinite chain.  Nevertheless, we will show, perhaps
surprisingly, that every finite partial order embeds into
$(L,\subword)$. In fact, this holds whenever $L$ is an unbounded
regular language. The latter requires two propositions that we shall
prove only later. Recall that a word $w\in\Sigma^+$ is called
\emph{primitive} if there is no $r\in\Sigma^+$ with $w=rr^+$.

\begin{proof}[Proof of Theorem~\ref{T-unbounded}]
  By Theorem~\ref{T-bounded}, we may assume
  that $L$ is unbounded.

  By Proposition~\ref{P-2} below, there are words $x,y,u,v\in\Sigma^*$
  with $|u|=|v|$, $uv$ primitive, and $x\{u,v\}^*y\subseteq L$. By
  Proposition~\ref{P-1} below, any finite partial order embeds into
  $(\{u,v\}^*,\subword)$ and therefore into $(x\{u,v\}^*y,\subword)$
  which is a substructure of $(L,\subword)$, i.e., every finite
  partial order embeds into $(L,\subword)$.

  Hence $\varphi=\exists x_1,x_2,\dots,x_n\colon\psi$ with $\psi$
  quantifier-free holds in $(L,\subword)$ iff it holds in some finite
  partial order whose size can be bounded by $n$. Since there are only
  finitely many such partial orders, the result follows.
\end{proof}

The first proposition used in the above proof deals with the existence
of certain primitive words for every unbounded regular language.

\begin{proposition}\label{P-2}
  For every unbounded regular language $L\subseteq \Sigma^*$, there
  are words $x,u,v,y\in \Sigma^*$ so that
  \begin{enumerate}
  \item $|u|=|v|$,
  \item the word $uv$ is primitive, and
  \item $x\{u,v\}^*y\subseteq L$.
  \end{enumerate}
\end{proposition}

\begin{proof}
  Since $L$ is unbounded and regular, there are words
  $x,y,p,q\in\Sigma^*$ with $|p|=|q|$, $p\neq q$, and
  $x\{p,q\}^*y\subseteq L$. Set $r=pq$ and $s=pp$.

  Then $|r|=|s|$ and $x\{r,s\}^*y\subseteq x\{p,q\}^*y\subseteq L$.
  Suppose $r$ and $s$ are conjugate. Since $s=p^2$, this implies
  $r=yxyx$ with $p=yx$, i.e., $r$ is the square of some word $yx$ of
  length $|p|=|q|$. But this contradicts $r=pq$ and $p\neq q$. Hence
  $r$ and $s$ are not conjugate.

  Next let $n=|r|$, $u=rs^{n-1}$ and $v=s^n$.

  By contradiction, we show that $uv$ is primitive.

  Since we assume $uv=rs^{2n-1}$ not to be primitive, there is a word
  $w\in\Sigma^*$ with $rs^{2n-1}\in ww^+$.  Observe that there is a
  $t\in\bN$ such that $n\le |w^t|\le n^2$: If $|w|\ge n$, we can
  choose $t=1$ since $|w|\le\frac{1}{2}|rs^{2n-1}|=n^2$ and if
  $|w|<n$, we can take $t=n$.

  Observe that $r$ and $w^t$ are prefixes of $uv=rs^{2n-1}$ of length
  $n$ and $\ge n$, respectively. Hence $r$ is a prefix of $w^t$.

  On the other hand, $v=s^n$ and $w^t$ are suffixes of $uv$ of length
  $n^2$ and $\le n^2$, respectively. Hence $w^t$ is a suffix of $v=s^n$.

  Taking these two facts together, we obtain that $r$ is a factor of
  $s^n$. Since $r$ and $s$ are not conjugate, this implies $pq=r=s=pp$ which
  contradicts $p\neq q$.
\end{proof}

The second proposition used above talks about the embeddability of
every finite partial order into certain regular languages of the form
$\{u,v\}^*$ where the words $u$ and $v$ originate from the previous
proposition. The proof of this embeddability requires a good deal of
preparation that deals with the combinatorics of subwords, more
precisely with the properties of ``prefix-maximal subwords''.

Let $x=a_1a_2\dots a_m$ and $y=b_1b_2\dots b_n$ with
$a_i,b_i\in \Sigma$. An \emph{embedding} of $x$ into $y$ is a mapping
$\alpha\colon \{1,2,\dots,m\}\to\{1,2,\dots,n\}$ with
$a_i=b_{\alpha(i)}$ and $i<j\iff\alpha(i)<\alpha(j)$ for all
$i,j\in\{1,2,\dots,m\}$. Note that $x\subword y$ iff there exists an
embedding of $x$ into $y$. This embedding is called \emph{initial} if
$\alpha(1)=1$, i.e., if the left-most position in $x$ hits the
left-most position in $y$. Symmetrically, the embedding $\alpha$ is
\emph{terminal} if $\alpha(m)=n$, i.e., if the right-most position in $x$
hits the right-most position in $y$.

We write $x\embedsright y$ if $x\subword y$ and every embedding of $x$
into $y$ is terminal. This is equivalent to saying that $x$, but no
word $xa$ with $a\in\Sigma$ is a subword of $y$. In other words,
$x\embedsright y$ if $x$ is a \emph{prefix-maximal subword of $y$}.

\begin{lemma}\label{L-initial-or-terminal}
  Let $w$ be primitive and $n>|w|$. Then, every embedding of $w^{n}$
  into $w^{n+1}$ is either initial or terminal.
\end{lemma}

\begin{proof}
  Let $\alpha$ be an embedding of $w^n$ into $w^{n+1}$ that is neither
  initial nor terminal. Consider the $n$ copies of $w$ in the word
  $w^n$.  We call such a copy \emph{gapless} if its image in $w^{n+1}$
  under $\alpha$ is contiguous.  Since the length difference between
  $w^n$ and $w^{n+1}$ is only $|w|<n$, there has to be at least one
  gapless copy of $w$, say the $i$th copy. The image of this copy is a
  contiguous subword of $w^{n+1}$ that spells $w$ and occurs at some
  position $i\cdot |w|+j$ with $j\in\{0,\ldots,|w|\}$. If $j=0$, then
  $\alpha$ is initial and if $j=|w|$, then $\alpha$ is terminal. This
  means $j\in\{1,\ldots,|w|-1\}$. However, since $w$ is primitive, it
  can occur as a contiguous subword in $w^{n+1}$ only at positions
  that are divisible by $|w|$, which is a contradiction.
\end{proof}

\begin{lemma}\label{L-congruence-property}
  The ordering $\embedsright$ is multiplicative: If
  $x,x',y,y'\in\Sigma^*$ with $x\embedsright y$ and $x'\embedsright
  y'$, then $xy\embedsright x'y'$.
\end{lemma}
\begin{proof}
  Suppose $xy\not\embedsright x'y'$. Since $xy\subword x'y'$, there is
  $a\in\Sigma$ such that $xya\subword x'y'$. Then, either $xb\subword
  x'$ where $b$ is the first letter of $ya$ (contradicting
  $x\embedsright x'$), or $ya\subword y'$ (contradicting
  $y\embedsright y'$).
\end{proof}

\begin{lemma}\label{L-right-embeddings}
  Let $u,v\in\Sigma^*$ be words such that $|u|=|v|$ and $uv$ is
  primitive. Then, for all $\ell,n\in\bN$ with $n>|uv|+\ell+2$, we
  have
  \begin{enumerate}[(i)]
  \item $(uv)^n\embedsright v(uv)^n$,
  \item $(uv)^\ell v(uv)^{n-\ell-1} \embedsright (uv)^n$, and
  \item $(uv)^{1+\ell} v(uv)^{n-\ell-2} \embedsright v(uv)^n$.
  \end{enumerate}
\end{lemma}
\begin{proof}
  For claim (i), suppose $\alpha$ is an embedding of $(uv)^n$ into
  $v(uv)^n$. Then $\alpha$ induces an embedding $\beta$ of
  $(uv)^{n-1}$ into $(uv)^n$. Note that $\beta$ cannot be initial
  because otherwise $\alpha$ would embed $uv$ into $v$. Thus, $\beta$
  is terminal by Lemma~\ref{L-initial-or-terminal}. Hence, $\alpha$ is
  terminal.

  For claim (ii), suppose $\alpha$ is an embedding of $(uv)^\ell v
  (uv)^{n-\ell-1}$ into $(uv)^n$.  Since $(uv)^n=(uv)^\ell
  (uv)^{n-\ell}$, $\alpha$ induces an embedding $\beta$ of
  $(uv)^{n-\ell-1}$ into $(uv)^{n-\ell}$. Again, $\beta$ cannot be
  initial because otherwise $\alpha$ would embed $(uv)^\ell v$ into
  $(uv)^\ell$. Therefore, $\beta$ is terminal according to
  Lemma~\ref{L-initial-or-terminal}, meaning that $\alpha$ is terminal as
  well.

  Finally, for claim (iii), suppose $\alpha$ is an embedding of
  $(uv)^{1+\ell}v(uv)^{n-\ell-2}$ into $v(uv)^n$. Since
  $v(uv)^n=v(uv)^{1+\ell}(uv)^{n-\ell-1}$, $\alpha$ induces an
  embedding $\beta$ of $(uv)^{n-\ell-2}$ into $(uv)^{n-\ell-1}$.
  Again, $\beta$ cannot be initial because otherwise, $\alpha$ would
  embed $(uv)^{1+\ell}v$ into $v(uv)^{1+\ell}$, but these are distinct
  words of equal length. Thus, Lemma~\ref{L-initial-or-terminal} tells us
  that $\beta$ must be terminal and hence also $\alpha$.
\end{proof}

\begin{proposition}\label{P-1}
  Let $u,v\in\Sigma^*$ be distinct such that $uv$ is primitive and
  $|u|=|v|$. Then every finite partial order embeds into
  $(\{u,v\}^*,\subword)$.
\end{proposition}

\begin{proof}
  For $m\in\bN$, let $\le$ denote the componentwise order on the set
  $\{0,1\}^m$ of $m$-dimensional vectors over $\{0,1\}$. Note that
  every finite partial order with $m$ elements embeds into
  $(\{0,1\}^m,\le)$. Hence, it suffices to embed this partial order
  into $(\{u,v\}^*,\subword)$.

  We define the map $\varphi_m\colon \{0,1\}^m\to \{u,v\}^*$ as
  follows.  Set $n=|uv|+m+3$. Then, for a tuple
  $t=(t_1,\ldots,t_m)\in \{0,1\}^m$, let
  \[
    \varphi_m(t_1,\ldots,t_m)=v^{t_1} (uv)^n \cdots v^{t_m} (uv)^n\,.
  \]

  It is clear that for $s,t\in\{0,1\}^m$, $s\le t$ implies
  $\varphi_m(s)\subword\varphi_m(t)$.

  Now let $s=(s_1,\dots,s_m)$ and $t=(t_1,\dots,t_m)$ be two vectors
  from $\{0,1\}^m$ with $\varphi_m(s)\subword\varphi_m(t)$.  Towards a
  contradiction, suppose $s\nleq t$. Then there is an $i\in[1,m]$ with
  $s_i=1$, $t_i=0$ and $s_j\le t_j$ for all $j\in[1,i-1]$.  Since
  $(uv)^n\embedsright v(uv)^n$ by Lemma~\ref{L-right-embeddings}(i)
  and clearly also $(uv)^n\embedsright (uv)^n$, we have
  $v^{s_j}(uv)^n\embedsright v^{t_j}(uv)^n$ for every $j\in[1,i-1]$.
  Furthermore, since $s_i=1$ and $t_i=0$, we have
  $v^{s_i}(uv)^{n-1}\embedsright v^{t_i}(uv)^n$ according to
  Lemma~\ref{L-right-embeddings}(ii) with $\ell=0$. Therefore,
  Lemma~\ref{L-congruence-property}(i) yields
  \begin{equation} v^{s_1}(uv)^n\cdots v^{s_{i}}(uv)^{n-1}\embedsright v^{t_1}(uv)^n\cdots v^{t_{i}}(uv)^n. \label{eq:non-embedding-base}\end{equation}
  We show by induction on $k$ that for every $k\in[i,n]$, there
  is an $\ell\in[1,k]$ with
  \begin{equation} v^{s_1}(uv)^n\cdots
    v^{s_{k}}(uv)^{n-\ell}\embedsright v^{t_1}(uv)^n\cdots
    v^{t_{k}}(uv)^n. \label{eq:non-embedding-hypothesis}
  \end{equation}
  Of course, \eqref{eq:non-embedding-base} is the base case. So let
  $k\in[i,n-1]$ and $\ell\in[1,k]$ such that
  \eqref{eq:non-embedding-hypothesis} holds. We distinguish three
  cases.
  \begin{enumerate}
  \item Suppose $s_{k+1}=0$.
    Then
   
    \[
      (uv)^n\embedsright v^{t_{k+1}}(uv)^n
    \]
    since either $t_{k+1}=0$ and the two words are the same, or
    $t_{k+1}=1$ and $(uv)^n\embedsright v(uv)^n$ by
    Lemma~\ref{L-right-embeddings}(i).  So together with the induction
    hypothesis \eqref{eq:non-embedding-hypothesis},
    Lemma~\ref{L-congruence-property}(i) yields
    \begin{align*}
      v^{s_1}(uv)^n\cdots v^{s_k}(uv)^n v^{s_{k+1}}(uv)^{n-\ell}
      &= v^{s_1}(uv)^n\cdots v^{s_k}(uv)^{n-\ell}\,
        (uv)^n && \\% \text{since }s_{k+1}=0\\
      &\embedsright v^{t_1}(uv)^n\cdots v^{t_k}(uv)^n\,
         v^{t_{k+1}}(uv)^n\,,
    \end{align*}
    where the equality is due to $s_{k+1}=0$.
    Since $\ell\in[1,k]\subseteq[1,k+1]$, this proves
    \eqref{eq:non-embedding-hypothesis} for $k+1$.
  \item Suppose $s_{k+1}=1$ and $t_{k+1}=0$.  By
    Lemma~\ref{L-right-embeddings}(ii), we have
    \[
       (uv)^\ell v(uv)^{n-(\ell+1)}
       \embedsright (uv)^n\,.
    \]
    So together with the induction hypothesis
    \eqref{eq:non-embedding-hypothesis},
    Lemma~\ref{L-congruence-property}(i) implies
    \begin{align*} 
      v^{s_1}(uv)^n\cdots v^{s_k}(uv)^n v^{s_{k+1}}(uv)^{n-(\ell+1)}
      &= v^{s_1}(uv)^n\cdots v^{s_k}(uv)^{n-\ell}\,
        (uv)^\ell v (uv)^{n-(\ell+1)}\\
      &\embedsright v^{t_1}(uv)^n\cdots v^{t_k}(uv)^n\,
         (uv)^n\\
      &= v^{t_1}(uv)^n\cdots v^{t_k}(uv)^n\,
         v^{t_{k+1}}(uv)^n\,,% &&\text{since }t_{k+1}=0\,.
    \end{align*}
    where the second equality is due to $t_{k+1}=0$.
    Since $\ell+1\in[1,k+1]$, this proves
    \eqref{eq:non-embedding-hypothesis} for $k+1$.
  \item If $s_{k+1}=1$ and $t_{k+1}=1$, then
    Lemma~\ref{L-right-embeddings}(iii) tells us that
    \[
      (uv)^\ell v(uv)^{n-(\ell+1)}\embedsright v(uv)^n\,.
    \]
    So together with the induction hypothesis
    \eqref{eq:non-embedding-hypothesis},
    Lemma~\ref{L-congruence-property}(i) implies
    \begin{align*}
      v^{s_1}(uv)^n\cdots v^{s_k}(uv)^n v^{s_{k+1}}(uv)^{n-(\ell+1)}
      &= v^{s_1}(uv)^n\cdots v^{s_k}(uv)^{n-\ell}\,
        (uv)^\ell v (uv)^{n-(\ell+1)}\\
      &\embedsright v^{t_1}(uv)^n\cdots v^{t_k}(uv)^n\,
         v(uv)^n\\
      &= v^{t_1}(uv)^n\cdots v^{t_k}(uv)^n\,
         v^{t_{k+1}}(uv)^n\,.
    \end{align*}
    Since $\ell+1\in[1,k+1]$, this proves
    \eqref{eq:non-embedding-hypothesis} for $k+1$.
  \end{enumerate}
  This completes the induction. Therefore, we have in particular
  \[
    v^{s_1}(uv)^n\cdots v^{s_m}(uv)^{n-\ell}
    \embedsright v^{t_1}(uv)^n\cdots v^{t_m}(uv)^{n}=\varphi_m(t)
  \]
  for some $\ell>0$. Since the left-hand side is a proper prefix of
  $\varphi_m(s)$, this contradicts $\varphi_m(s)\subword\varphi_m(t)$.
\end{proof}

This completes the proof of the main result of this section, i.e., of
Theorem~\ref{T-unbounded}.

\section{The $\Sigma_1$-theory with constants}

By Theorem~\ref{T-unbounded}, the $\Sigma_1$-theory of $(L,\subword)$
is decidable for all regular languages $L$. If $L$ is bounded, then
even the $\Sigma_1$-theory of $(L,\subword,(w)_{w\in L})$ is decidable
(Theorem~\ref{T-bounded}). This result does not extend to all regular
languages since, e.g., the $\Sigma_1$-theory of
$(\Sigma^*,\subword,(w)_{w\in\Sigma^*})$ is
undecidable~\cite{HalSZ17}. In this section, we present another class
of regular languages $L$ (besides the bounded ones) such that
\[
  \cS=(L,\subword,(w)_{w\in L})
\]
has a decidable $\Sigma_1$-theory.

Let $L\subseteq\Sigma^*$ be some language. Then \emph{almost all words
  from $L$ have a non-negligible number of occurrences of every
  letter} if there exists a positive real number $\varepsilon$ such
that for all $a\in\Sigma$ and all but finitely many words $w\in L$, we
have
\[
  \frac{|w|_a}{|w|}>\varepsilon\,.
\]
An example of such a regular language is $\{ab,ba\}^*$ (this class
contains all finite languages, is closed under union and concatenation
and under iteration, provided every word of the iterated language
contains every letter).

For $w\in\Sigma^*$, let $w\mathord\uparrow$ denote the set of
superwords of $w$, i.e., the upward closure of $\{w\}$ in
$(\Sigma^*,\subword)$.

The basic idea is, as in the proof of Theorem~\ref{T-unbounded}, to
embed every finite partial order into $(L,\subword)$. The following
lemma refines this embedability. Furthermore, it shows that
$L\setminus w\mathord{\uparrow}$ is finite in this case.

\begin{lemma}\label{L-embeds-above-w}
  Let $L\subseteq\Sigma^*$ be an unbounded regular language such that
  almost all words from $L$ have a non-negligible number of
  occurrences of every letter. Let $w\in\Sigma^*$. Then every finite
  partial order $(P,\le)$ can be embedded into
  $(L\cap w\mathord{\uparrow},\subword)$. Furthermore, the set
  $L\setminus w\mathord{\uparrow}$ is finite.
\end{lemma}

Note that $L\cap w\mathord{\uparrow}$ is regular, but not necessarily
unbounded (it could even be finite). Hence the first claim is not an
obvious consequence of Propositions~\ref{P-1} and \ref{P-2}.

\begin{proof}
  Since $L$ is regular and unbounded, there are words
  $x,u,v,y\in\Sigma^*$ with $|u|=|v|>0$, $uv$ primitive, and
  $x\{u,v\}^*y\subseteq L$ (by Proposition~\ref{P-2}). In particular,
  $xu^*y\subseteq L$. Let $a\in\Sigma$ and suppose $|u|_a=0$. Then
  \[
    \lim_{n\to\infty}\frac{|xu^ny|_a}{|xu^ny|}=0
  \]
  contradicting that almost all words from $L$ have a non-negligible
  number of occurrences of every letter. Hence, $u$ contains every
  letter from $\Sigma$ implying $w\subword u^{|w|}$. Set
  $x'=x u^{|w|}$. Then we have
  $x'\{u,v\}^*y\subseteq L\cap w\mathord{\uparrow}$. From
  Proposition~\ref{P-1}, we learn that $(P,\le)$ can be embedded into
  $(\{u,v\}^*,\subword)$. Hence, it can be embedded into
  $(x'\{u,v\}^*y,\subword)$ and therefore into
  $(L\cap w\mathord{\uparrow},\subword)$.

  Next, we show that $L\setminus w\mathord{\uparrow}$ is finite. Let
  $M=(Q,\Sigma,\iota,\delta,F)$ be the minimal deterministic finite
  automaton accepting $L$. Let $v\in L$ with $|v|\ge|Q|\cdot|w|$. Then
  we can factorize the word $v$ into $v=v_0v_1\cdots v_{|w|+1}$ such
  that $\delta(\iota,v_0)=\delta(\iota,v_0v_1\cdots v_i)$ for all
  $1\le i\le|w|$. With $q=\delta(\iota,v_0)$, we obtain
  $\delta(q,v_i)=q_i$ for all $1\le i\le |w|$ and therefore
  $v_0v_1\cdot v_{i-1}v_i^*v_{i+1}\cdots v_{|w|+1}\subseteq L$.  Since
  almost all words from $L$ have a non-negligible number of
  occurrences of every letter, this implies (as above) that $v_i$
  contains all letters from $\Sigma$. Since this holds for all
  $1\le i\le|w|$, we obtain $w\subword v_1v_2\dots v_{|w|}\subword v$
  and therefore $v\notin L\setminus w\mathord{\uparrow}$. Hence,
  indeed, $L\setminus w\mathord{\uparrow}$ is finite.
\end{proof}

\begin{theorem}\label{T-nice}
  Let $L\subseteq\Sigma^*$ be an unbounded regular language such that
  almost all words from $L$ have a non-negligible number of
  occurrences of every letter. Then the $\Sigma_1$-theory of
  $(L,\subword,(w)_{w\in L})$ is decidable.
\end{theorem}

\begin{proof}
  We want to show that satisfiability in $(L,\subword,(w)_{w\in L})$
  is decidable for quantifier-free formulas, i.e., for positive
  Boolean combinations $\varphi$ of literals of the following forms
  (where $x$ and $y$ are arbitrary variables and $w$ an arbitrary word
  from $L$):
  
  \begin{minipage}[t]{.2\linewidth}
    \begin{enumerate}[(i)]
    \item\label{form:xw} $x\subword w$
    \item\label{form:xnw} $x\not\subword w$
    \end{enumerate}
  \end{minipage}
  \begin{minipage}[t]{.2\linewidth}
    \begin{enumerate}[(i)]
      \setcounter{enumi}{2}
    \item\label{form:wx} $w\subword x$
    \item\label{form:wnx} $w\not\subword x$
    \end{enumerate}
  \end{minipage}
  \begin{minipage}[t]{.2\linewidth}
    \begin{enumerate}[(i)]
      \setcounter{enumi}{4}
    \item\label{form:xy} $x\subword y$
    \item\label{form:xny} $x\not\subword y$
    \end{enumerate}
  \end{minipage}
  
  %% \begin{minipage}[t]{.2\linewidth}
  %%   \begin{enumerate}[(i)]
  %%     \setcounter{enumi}{6}
  %%   \item\label{form:equal} $x=w$
  %%   \end{enumerate}
  %% \end{minipage}
  \noindent
  Note that literals of the form $x=y$ can be
  written as $x\subword y\land y\subword x$, $x\neq y$ as
  $x\not\subword y\lor y\not\subword x$, and similarly $x\neq w$ as
  $x\not\subword w\lor w\not\subword x$. Furthermore, literals
  mentioning two words like $u\subword v$ can be replaced by
  $\top=(x\subword y\lor x\not\subword y)$ or $\bot=(x\subword y\land
  x\not\subword y)$. By bringing the formula in disjunctive normal
  form, we may assume that we are given a disjunction of conjunctions
  of such literals.

  \emph{Step 1.} We first show that literals of types \eqref{form:xw} and
  \eqref{form:wnx} can be eliminated. To this end, observe that for
  each $w\in L$, both of the sets 
  \begin{align*} \{u\in L\mid u\subword w\}, && \{u\in L\mid w\not\subword u\} \end{align*}
  are finite.  In the case $\{u\in L\mid u\subword w\}$, this is
  trivial. In the case of $\{u\in L\mid w\not\subword u\}$, this is
  the second claim in Lemma~\ref{L-embeds-above-w}.  Thus, every
  conjunction that contains a literal $x\subword w$ or $w\not\subword
  x$, constrains $x$ to finitely many values. Therefore, we can
  replace this conjunction with a disjunction of conjunctions that
  result from replacing $x$ by one of these values.  (Here, we might
  obtain literals $u\subword v$ or $u\not\subword v$, but those can be
  replaced by $\bot$ and $\top$ as above).

  Note that such a replacement reduces the number of variables by one.
  We repeat this replacement until there are no more literals of the
  form \eqref{form:xw} and \eqref{form:wnx}. Since we replace each
  conjunction with (a disjunction of) conjunctions that have fewer
  variable, this has to terminate. Thus, we arrive at a disjunction of
  conjunctions of literals of the forms
  \eqref{form:xnw},\eqref{form:wx},\eqref{form:xy}, and
  \eqref{form:xny}.

  \emph{Step 2.} In the second step, we will eliminate literals of the
  form \eqref{form:xnw}.  Note that the language
  $\{u\in L\mid u\not\subword w\}$ is upward closed in
  $(L,\subword)$. Since $L$ is regular, we can compute the finite set
  of minimal elements of this set. Thus, $x\not\subword w$ is
  equivalent to a finite disjunction of literals of the form
  $w'\subword x$.  As a result, we get a positive Boolean combination
  $\psi$ of literals of the form \eqref{form:wx}, \eqref{form:xy},
  \eqref{form:xny} that is equivalent to $\varphi$.

  \emph{Step 3.} In the third step, we check whether our formula is satisfiable.
  We may assume that $\psi$ is in disjunctive normal
  form. To verify whether $\psi$ is satisfiable in $(L,\subword)$, it
  therefore suffices to verify satisfiability of conjunctions of
  literals of the form \eqref{form:wx}, \eqref{form:xy},
  \eqref{form:xny}.
  So let $\gamma$ be such a conjunction.
  %% If $\gamma$ contains literals
  %% $x=u$ and $x=v$ for two distinct words $u$ and $v$, then $\gamma$ is
  %% not satifiable in $(L,\subword)$. Otherwise, for all literals $x=w$
  %% in $\gamma$, replace all occurrences of $x$ in $\gamma$ by $w$, and
  %% delete all literals $x=w$ from $\gamma$. The resulting conjunction
  %% of literals of the form \eqref{form:wx}, \eqref{form:xy}, and
  %% \eqref{form:xny} is equivalent to $\gamma$.
  It can be written as
  $\gamma_1\land\gamma_2$ where $\gamma_1$ is a conjunction of
  literals of the form \eqref{form:wx} and $\gamma_2$ is a conjunction
  of literals of the form \eqref{form:xy} and \eqref{form:xny}.

  Let $n$ denote the number of variables appearing in $\gamma_2$. If
  $\gamma$ is satisfiable in $(L,\subword)$, then $\gamma_2$ is
  satisfied by some partial order with at most $n$
  elements. Conversely, let $\gamma_2$ be satisfied by $(P,\le)$ where
  $P$ has at most $n$ elements. Let, furthermore, $w$ denote some
  concatenation of all words $w$ appearing in the formula $\gamma_1$. By
  the first claim of Lemma~\ref{L-embeds-above-w}, the finite partial
  order $(P,\le)$ can be embedded into
  $(L\cap w\mathord{\uparrow},\subword)$. Consequently,
  $\gamma_1\land\gamma_2$ is satisfiable in
  $(L\cap w\mathord{\uparrow},\subword)$ and therefore in
  $(L,\subword)$. In summary, $\gamma$ is satifiable in $(L,\subword)$
  iff $\gamma_2$ holds in some finite partial order of size at most
  $n$. Since there are only finitely many such finite partial orders,
  we get that satisfiability of $\gamma$ in $(L,\subword)$ is
  decidable.
  %
  %
  % First note that $\{v\in L\mid v\subword w\}$ is finite and
  % computable from $w$. Hence the literal $x\subword w$ is equivalent
  % to a finite disjunction of literals of the form $x=w'$.
  %
  % Finally, by the second claim in Lemma~\ref{L-embeds-above-w}, also
  % the literal $w\not\subword x$ can be replaced by a finite
  % disjunction of literals of the form $x=w'$.
\end{proof}

\section*{Open questions}

We did not consider complexity issues. In particular, from
\cite{KarS15}, we know that the $\FO^2$-theory of the structure
$(\Sigma^*,\subword,(w)_{w\in\Sigma^*})$ can be decided in elementary
time. We currently work out the details for the extension of this
result to the $\CMOD^2$-theory of the structure
$(L,\subword,(w)_{w\in L})$ for $L$ regular. We reduced the
$\FOMOD$-theory of the full structure (for $L$ context-free and
bounded) to the $\FOMOD$-theory of $(\bN,+)$ which is known to be
decidable in elementary time \cite{HabK15}. Unfortunately, our
reduction increases the formula exponentially due to the need of
handling statements of the form ``there is an even number of pairs
$(x,y)\in\bN^2$ such that ...'' It should be checked whether the proof
from \cite{HabK15} can be extended to handle such statements in
$\FOMOD$ for $(\bN,+)$ directly.

Finally, we did not give any new undecidability results. For example,
we know that the $\Sigma_1$-theory of $(L,\subword,(w)_{w\in L})$ is
undecidable for $L=\Sigma^*$ \cite{HalSZ17} and decidable for
$L=\{ab,ba\}^*$ (Theorem~\ref{T-nice}). To narrow the gap between
decidable and undecidable cases, one should find more undecidable
cases.

\end{document}